%% file: main.tex
\documentclass[a4paper,USenglish,cleveref]{lipics-v2021}

\input{0_packages}

\input{0_macros}
\nolinenumbers

\title{Concurrent Stochastic Games with Stateful-discounted and Parity Objectives: Complexity and Algorithms}

\titlerunning{Concurrent Stochastic Games with Stateful-discounted and Parity Objectives}

\author{Ali Asadi}{Institute of Science and Technology Austria (ISTA), Austria }{ali.asadi@ista.ac.at}{https://orcid.org/0009-0005-2839-953X}{}
\author{Krishnendu Chatterjee}{Institute of Science and Technology Austria (ISTA), Austria }{krishnendu.chatterjee@ista.ac.at}{https://orcid.org/0000-0002-4561-241X}{}
\author{Raimundo Saona}{Institute of Science and Technology Austria (ISTA), Austria }{raimundojulian.saonaurmeneta@ista.ac.at}{https://orcid.org/0000-0001-5103-038X}{}
\author{Jakub Svoboda}{Institute of Science and Technology Austria (ISTA), Austria }{jakub.svoboda@ista.ac.at}{https://orcid.org/0000-0002-1419-3267}{}

\authorrunning{A. Asadi, K. Chatterjee, R. Saona, and J. Svoboda}

\Copyright{Ali Asadi, Krishnendu Chatterjee, Raimundo Saona, and Jakub Svoboda}

\ccsdesc{Theory of computation~Logic and verification}
\keywords{Concurrent Stochastic Games, Parity Objectives, Discounted-sum Objectives}

\funding{This research was partially supported by ERC CoG 863818 (ForM-SMArt), Austrian Science Fund (FWF) 10.55776/COE12, and French Agence Nationale de la Recherche (ANR) ANR-21-CE40-0020 (CONVERGENCE project).}

\begin{document}

\maketitle

\input{./abstract.tex}

\newpage
\input{./nomenclature.tex}

\input{./introduction.tex}

\input{./preliminaries.tex}
\input{./mathematical_properties.tex}

\input{./algorithms.tex}
\input{./complexity.tex}

\input{./conclusion.tex}

\newpage
\bibliographystyle{plain}
\bibliography{refs.bib}



\end{document}

%% file: 0_packages.tex
\usepackage[notocbasic]{nomencl} 
\makenomenclature

\usepackage{svg}

\usepackage{bm}

\usepackage{tcolorbox}

\usepackage{mathtools}

\usepackage{algorithmicx,algorithm}
\usepackage[noend]{algpseudocode}

%% file: 0_macros.tex
\newcommand{\D}{\mathcal{D}}
\newcommand{\B}{\mathrm{B}}
\newcommand{\F}{\mathcal{F}}
\newcommand{\setPolynomials}{\mathcal{P}}
\newcommand{\setPolynomialDenominators}{\mathcal{Q}}
\newcommand{\calO}{\mathcal{O}}
\newcommand{\tilO}{\widetilde{\mathcal{O}}}

\newcommand{\eps}{\varepsilon}

\newcommand{\defas}{\coloneqq}
\newcommand{\sadef}{\eqqcolon}

\newcommand{\PP}{\mathbb{P}}
\newcommand{\EE}{\mathbb{E}}
\newcommand{\ZZ}{\mathbb{Z}}
\newcommand{\RR}{\mathbb{R}}

\newcommand{\MDP}{\mathscr{P}}
\newcommand{\MC}{\mathscr{C}}

\newcommand{\support}{\text{supp}}
\newcommand{\payoff}{\nu}
\newcommand{\reward}{r}
\newcommand{\discfac}{{\Lambda}}
\newcommand{\distribution}{\mu}
\newcommand{\discountfactor}{\lambda}
\newcommand{\assign}{\chi}

\newcommand{\deter}{\nabla}
\newcommand{\id}{\text{Id}}
\newcommand{\parind}{d}
\newcommand{\degree}{D}

\newcommand{\prob}{\delta} 
\newcommand{\statenum}{n}
\newcommand{\actionnum}{m}
\newcommand{\Actionone}{\mathcal{A}} 
\newcommand{\Actiontwo}{\mathcal{B}} 
\newcommand{\actionone}{a}
\newcommand{\actiontwo}{b}
\newcommand{\States}{\mathcal{S}} 
\newcommand{\state}{s} 
\newcommand{\Plays}{\Omega} 
\newcommand{\play}{\omega} 

\newcommand{\Target}{T}
\newcommand{\strategyone}{\sigma}
\newcommand{\strategytwo}{\tau}

\newcommand{\Strategyone}{\Sigma}
\newcommand{\Strategytwo}{\Gamma}

\newcommand{\prior}{p}

\newcommand{\UndiscVal}{{\sc LimitValue}}
\newcommand{\ParVal}{{\sc ParityValue}}
\newcommand{\ApproxDiscounted}{{\sc ApproxDiscounted}}
\newcommand{\ApproxLimit}{{\sc ApproxLimit}}

\newcommand{\Det}{{\sc Det}}
\newcommand{\Val}{{\sc Val}}

\DeclareMathOperator{\bit}{bit}
\DeclareMathOperator{\sign}{sign}
\DeclareMathOperator{\rel}{rel} 
\DeclareMathOperator{\reach}{Reach}
\DeclareMathOperator{\parity}{Parity}
\DeclareMathOperator{\discounted}{Disc}
\DeclareMathOperator{\len}{len}

\DeclareMathOperator{\val}{val}

\Crefname{equation}{Eq.}{Eqs.}

%% file: abstract.tex
\begin{abstract}
We study two-player zero-sum concurrent stochastic games with finite state and action space played for an infinite number of steps.
In every step, the two players simultaneously and independently choose an action.
Given the current state and the chosen actions, the next state is obtained according to a stochastic transition function.
An objective is a measurable function on plays (or infinite trajectories) of the game, and the value for an objective is the maximal expectation that the player can guarantee against the adversarial player.
We consider: (a)~stateful-discounted objectives, which are similar to the classic discounted-sum objectives, but states are associated with different discount factors rather than a single discount factor; and (b)~parity objectives, which are a canonical representation for $\omega$-regular objectives.
For stateful-discounted objectives, given an ordering of the discount factors, the limit value is the limit of the value of the stateful-discounted objectives, as the discount factors approach zero according to the given order.

The computational problem we consider is the approximation of the value within an arbitrary additive error.
The above problem is known to be in EXPSPACE for the limit value of stateful-discounted objectives and in PSPACE for parity objectives.
The best-known algorithms for both the above problems are at least exponential time, with an exponential dependence on the number of states and actions.
Our main results for the value approximation problem for the limit value of stateful-discounted objectives and parity objectives are as follows: (a)~we establish TFNP[NP] complexity; and (b)~we present algorithms that improve the dependency on the number of actions in the exponent from linear to logarithmic.
In particular, if the number of states is constant, our algorithms run in polynomial time.
\end{abstract}

%% file: nomenclature.tex
\renewcommand{\nomname}{List of major symbols}
\setlength{\nomlabelwidth}{4em}

\renewcommand\nomgroup[1]{%
  \item[\bfseries
  \ifstrequal{#1}{A}{Main symbols}{
  \ifstrequal{#1}{B}{Dynamics}{
  \ifstrequal{#1}{C}{Objectives and values}{
  \ifstrequal{#1}{O}{Other symbols}{
  }}}}
]}

\nomenclature[A, 01]{\(\statenum\)}{Number of states}
\nomenclature[A, 02]{\(\actionnum\)}{Number of actions per player and state}
\nomenclature[A, 03]{\(\parind\)}{Index of parity and number of discount factors}
\nomenclature[A, 04]{\(\B\)}{Bit size of transition and reward functions}
\nomenclature[A, 05]{\(\eps\)}{Additive error}


\nomenclature[C]{\(\val_\discfac\)}{Stateful-discounted value}
\nomenclature[C]{\(\val_\assign\)}{Limit of the stateful-discounted value}
\nomenclature[C]{\(\val_\Target\)}{Reachability value}
\nomenclature[C]{\(\reach_\Target\)}{Reachability objective}
\nomenclature[C]{\(\val_\prior\)}{Parity value}
\nomenclature[C]{\(\val\)}{Value of a matrix game}
\nomenclature[C]{\(\parity_\prior\)}{Parity objective}
\nomenclature[C]{\(\discounted_\discfac\)}{Stateful-discounted objective}
\nomenclature[C]{\(\reward\)}{Reward function}

\nomenclature[C]{\(\prior\)}{Priority function}
\nomenclature[C]{\(\discfac\)}{Discount function}
\nomenclature[C]{\(\assign\)}{Assignment function}

\nomenclature[C]{\(\Target\)}{Set of targets}
\nomenclature[C]{\(\discountfactor\)}{Discount factor}


\nomenclature[B]{\(\States\)}{Set of states}
\nomenclature[B]{\(\Actionone\)}{Set of actions per state of Player 1}
\nomenclature[B]{\(\Actiontwo\)}{Set of actions per state of Player 2}
\nomenclature[B]{\(\prob\)}{Transition probability function}
\nomenclature[B]{\(\strategyone\)}{Strategy of Player 1}
\nomenclature[B]{\(\strategytwo\)}{Strategy of Player 2}
\nomenclature[B]{\(\Strategyone\)}{Set of strategies of Player 1}
\nomenclature[B]{\(\Strategytwo\)}{Set of strategies of Player 2}
\nomenclature[B]{\(\MDP\)}{Markov Decision Process (MDP)}
\nomenclature[B]{\(\MC\)}{Markov Chain (MC)}
\nomenclature[B]{\(G\)}{Concurrent Stochastic Game (CSG)}
\nomenclature[B]{\(G_\strategyone, G_\strategytwo\)}{Induced MDPs}
\nomenclature[B]{\(G_{\strategyone, \strategytwo}\)}{Induced MC}


\nomenclature[O]{\(\F\)}{Set of floating point numbers}
\nomenclature[O]{\(\D\)}{Set of distributions represented in floating point}
\nomenclature[O]{\(\setPolynomials, \setPolynomialDenominators\)}{Sets of polynomials}
\nomenclature[O]{\(\degree\)}{Degree of a polynomial}


\nomenclature[O]{\(\distribution\)}{Probability distribution}
\nomenclature[O]{\(\deter\)}{Determinant function for stationary strategies}
\nomenclature[O]{\(\id\)}{Identity matrix}
\nomenclature[O]{\(P, Q, C\)}{Polynomials}
\nomenclature[O]{\(W\)}{Auxiliary matrix game}


\printnomenclature

%% file: introduction.tex
\section{Introduction}
In this work, we present improved complexity results and algorithms 
for the value approximation of concurrent stochastic games with two 
classic objectives. 
Below we present the model of concurrent stochastic games, the relevant objectives,
the computational problems, previous results, and finally our contributions.

\smallskip\noindent{\bf Concurrent stochastic games.}
Concurrent stochastic games are two-player zero-sum games played on finite-state graphs 
for an infinite number of steps. 
These games were introduced in the seminal work of Shapley~\cite{shapley1953stochastic} and are a fundamental model in game theory.
In each step, both players simultaneously and independently of the other player choose an action.
Given the current state and the chosen actions, the next state is obtained according to a stochastic 
transition function.
An infinite number of such steps results in a play which is an infinite sequence of states and actions. 
Concurrent stochastic games have been widely studied in the literature from the mathematical perspective~\cite{shapley1953stochastic,everett1957RecursiveGames,FV97,MN81}, and from the algorithmic and computational complexity perspective, including: complexity for reachability objectives~\cite{dealfaro1998ConcurrentReachabilityGames,etessami2008RecursiveConcurrentStochastic,frederiksen2013approximating,hansen2009WinningConcurrentReachability}, algorithms for limit-average objectives~\cite{hansen2011exact,oliu2021new}, complexity for qualitative solutions for omega-regular objectives~\cite{dealfaro2000ConcurrentOmegaregularGames}, complexity for quantitative solutions for omega-regular objectives~\cite{chatterjee2006complexity,de2001quantitative}, and in the context of temporal logic~\cite{alur2002alternating}. 
In particular, in the analysis of reactive systems, concurrent games provide the appropriate model for reactive systems with components that interact synchronously~\cite{alur2002alternating,de2000control,de2001control}.

\smallskip\noindent{\bf Objectives.}
An objective is a measurable function that assigns to every play a real-valued reward.
The classic discounted-sum objective is as follows: every transition is assigned
a reward and the objective assigns to a play the discounted-sum of the rewards. 
While the classic objective has a single discount factor, the stateful-discounted objective has multiple discount factors. 
In the stateful-discounted objective, each state is associated with a discount factor, and, in the objective, the discount at a step depends on the current state.
We also consider the boolean parity objectives, which are a canonical form to express all $\omega$-regular objectives~\cite{thomas1997languages}, which provide a robust specification for all properties that arise in verification.
For example, all LTL formulas can be converted to deterministic parity automata.
In parity objectives, every state is associated with an integer priority, and a play is winning for (or satisfies) the objective if the minimum priority visited infinitely often is even.

\smallskip\noindent{\bf Strategies and values.}
Strategies are recipes that define the choice of actions of the players.
They are functions that, given a game history, return a distribution over actions.
Given a concurrent stochastic game and an objective, the value of Player~1 at a state is the maximal expectation that the player can guarantee for the objective against all strategies of Player~2.
For stateful-discounted objectives, given an ordering of the discount factors, the limit value at a state is the limit of the value function for the discounted objective as the discount factors approach zero in the given order.

\smallskip\noindent{\bf Computational problems.}
Given a concurrent stochastic game, the main computational problems are:
(a)~the {\em value-decision} problem, given a state and a threshold $\alpha$, asks whether the value at the state is at least $\alpha$;
and (b)~the {\em value-approximation} problem, given a state and an error $\eps>0$, asks to compute an approximation of the value for the state within an additive error of $\eps$.
We consider the above problems for the limit value of stateful-discounted objectives and the value for parity objectives.

\smallskip\noindent{\bf Motivation.} 
The motivation to study the limit of the stateful-discounted objective is as follows. 
First, this limit generalizes the classic limit-average objectives. 
Second, it characterizes the value for the parity objectives in concurrent stochastic games~\cite{gimbert2005discounting,de2003discounting}, where the order of limit corresponds to the order of importance of priorities in parity objectives. 
Third, the limit value has been shown to correspond to the value for other objectives such as priority mean-payoff for various subclasses of concurrent stochastic games~\cite{gimbert2012blackwell}. 
The motivation to study the value-approximation problem as opposed to the value-decision is that for concurrent games, even for special classes of objectives such as reachability and safety, values can be irrational, and the decision problem related to exact value is SQRT-SUM hard~\cite{etessami2008RecursiveConcurrentStochastic} as explained below. 
Hence, approximation of values is a natural problem to study from an algorithmic and computational complexity perspective.

\smallskip\noindent{\bf Previous results.} 
For a single discount factor, the limit value corresponds to the value of the well-studied mean-payoff or long-run average objectives~\cite{MN81}, and, for parity objectives, the computational problems admit a linear reduction to the limit value of stateful-discounted objectives~\cite{gimbert2005discounting,de2003discounting}. 
The value-decision problem for concurrent stochastic games is SQRT-SUM hard~\cite{etessami2008RecursiveConcurrentStochastic}:
this result holds for reachability objectives, and hence also for parity objectives and the limit value for even a single discount factor.
The SQRT-SUM problem is a classic problem in computational geometry, and whether SQRT-SUM belongs to NP has been a long-standing open problem. 
The complexity upper bounds for the value-approximation problem of concurrent stochastic games are as follows:
(a)~EXPSPACE for the limit value of stateful-discounted objectives; and 
(b)~PSPACE for parity objectives~\cite{chatterjee2007stochastic,chatterjee2006complexity}.
The above result for the limit value follows from a reduction to the theory of reals,
where the number of discount factors corresponds to the number of quantifier alternation.
For the special class of reachability objectives, the complexity upper bound of TFNP[NP] for 
the value-approximation problem has been established in~\cite{frederiksen2013approximating}, where TFNP[NP] is the total functional form of the second level of the polynomial hierarchy. 
The result of \cite{frederiksen2013approximating} has been recently extended to limit-average objectives (which correspond to the limit-value of single discount factor)~\cite{bose2023bounded}.
To the best of our knowledge, the above complexity upper bounds are the best bounds for 
limit value of general stateful-discounted objectives and parity objectives.
The best-known algorithms for the value-approximation problem are as follows:
(a)~double exponential time for the limit value of stateful-discounted objectives;
(b)~exponential time for parity objectives, where the exponent is a product that depends at least linearly on the number of states and actions~\cite{chatterjee2006complexity,chatterjee2007stochastic} (see \Cref{Section: Previous and our results} for further details).
While iterative approaches are desirable, they neither exist for parity objectives nor guarantee efficiency even in special cases.
For example, for reachability and safety objectives, iterative approaches like value-iteration or strategy-iteration have a double-exponential lower bound~\cite{hansen2011ComplexitySolvingReachability}; and, for parity objectives, iterative approaches like strategy iteration are not known as strategies require infinite-memory~\cite{dealfaro2000ConcurrentOmegaregularGames}.

\smallskip\noindent{\bf Our contributions.}
In this work, our main contributions are as follows:
(a)~we establish TFNP[NP] upper bounds for the value-approximation problem for concurrent stochastic games, both for the limit value of stateful-discounted objectives and the parity objectives; and (b)~we present algorithms which are exponential time and  improve the dependency on the number of actions in the exponent from linear to logarithmic.
In particular, if the number of states is constant, our algorithms run in polynomial time.
The comparison of previous results and our results is summarized in \Cref{Table: Complexity result comparison,Table: Algorithms result comparison}.
\begin{table}
    \vspace{-1em}
    \centering
    \begin{tabular}{|c|c|c|}
        \hline
        & \multicolumn{2}{c|}{\textbf{Complexity}} \\
        \cline{2-3}
        & \textbf{Previous} & \textbf{Ours} \\
        \hline
        \multirow{2}{*}{\textbf{Limit}} & EXPSPACE &  \\
        & (Theory of reals) & TFNP[NP] \\
        \cline{1-2}
        \multirow{2}{*}{\textbf{Parity}}& PSPACE &  (\Cref{Result: Computational result for limit value}-\Cref{Item: Complexity class for limit value}, \Cref{Result: Computational result for parity value}-\Cref{Item: Complexity class for parity value}) \\ 
        & \cite{chatterjee2006complexity,chatterjee2007stochastic} & \\
        \hline
    \end{tabular}
    \caption{Complexity upper bounds of the value-approximation in concurrent stochastic games for the limit value of stateful-discounted objectives and parity objectives.}
    \label{Table: Complexity result comparison}
\end{table} 

\begin{table}
    \centering
    \begin{tabular}{| c | c | c |}
        \hline
        & \multicolumn{2}{c|}{\textbf{Algorithms}} \\
        \cline{2-3}
        & \textbf{Previous} & \textbf{Ours} \\
        \hline
        \multirow{2}{*}{\textbf{Limit}} & {\footnotesize	$\exp \left (\calO(2^\parind \actionnum^2\statenum + \log(1/\eps) +   \log (\B)) \right )$} & \multirow{5}{*}{
            \shortstack{
                {\footnotesize	$\exp \left( \calO \left( 
                    \begin{aligned}
                        &\statenum\parind \log(\actionnum) + \log(\B)\\ 
                        &\quad+ \log(\log(1/\eps)) 
                    \end{aligned}
                \right) \right)$} \\
                (\Cref{Result: Computational result for limit value}-\Cref{Item: Algorithmic time for limit value}, \\
                \Cref{Result: Computational result for parity value}-\Cref{Item: Algorithmic time for parity value})
            }
        } \\
        & (Theory of reals) &  \\
        \cline{1-2}
        \multirow{3}{*}{\textbf{Parity}}
            & \multirow{3}{*}{
                \shortstack{
                    {\footnotesize $\exp \left( \calO \left( 
                        \begin{aligned}
                            &\actionnum \statenum + \parind \log(\statenum) + \log(\B) \\
                            &\quad+ \log(\log(1/\eps)) 
                        \end{aligned}
                    \right) \right)$
                    }\\\cite{chatterjee2006complexity,chatterjee2007stochastic}
                }
            } 
            & \\ 
        &  &  \\
        &  &  \\
        \hline
    \end{tabular}
    \caption{Algorithmic upper bounds of the value-approximation in concurrent stochastic games for the limit value of stateful-discounted objectives and parity objectives, where $\statenum$ is the number of states, $\actionnum$ is the number of actions, $\parind$ is the number of discount factors/parity index, $\B$ is the bit-size of numbers in the input, $\eps$ is the additive error, and $\exp$ is the exponential function. }
    \label{Table: Algorithms result comparison}
\end{table} 

\smallskip\noindent{\bf Technical contributions.} 
Our main technical contributions are as follows. We first present a bound on the roots of multi-variate polynomials with integer coefficients~(\Cref{Section: Bounds on roots of polynomials}).
Given the bounds on roots of polynomials, we establish new characterizations for the limit and stateful-discounted values~(\Cref{Section: Characterization of discounted value} and \Cref{Section: Characterization of limit value}), which lead to an approximation of the limit value by the stateful-discounted value when the discount factors are double-exponentially small (\Cref{Section: Approximation of limit value}). 
Given this connection, we establish the improved complexities and algorithms for the value-approximation for the limit value of stateful-discounted objectives and parity objectives in \Cref{Section: Algorithms} and \Cref{Section: Complexities}.

%% file: preliminaries.tex
\section{Preliminaries}
\label{Section: Preliminaries}
We present standard definitions related to concurrent stochastic games.

\smallskip\noindent{\bf Basic Notations.} Given a finite set $\mathcal{X}$, a probability distribution over $\mathcal{X}$ is a function $\distribution \colon \mathcal{X} \to [0, 1]$ such that $\sum_{x \in \mathcal{X}} \distribution(x) = 1$. 
The set of all probability distributions over $\mathcal{X}$ is denoted by $\Delta(\mathcal{X})$.
For $\distribution \in \Delta(\mathcal{X})$, the support of $\distribution$ is defined as $\support(\distribution) \defas \{x \mid \distribution(x) > 0\}$. 
For a positive integer $k$, the set of positive integers smaller than or equal to $k$ is defined as $[k] \defas \{ 1, \ldots, k \}$. 
Given a real $x$, we denote $2^x$ by $\exp(x)$.

\smallskip\noindent{\bf Concurrent stochastic games.} 
A concurrent stochastic game (CSG) is a two-player finite game 
$G = (\States, \Actionone, \Actiontwo, \prob)$ consisting of
\begin{itemize}
    \item the set of states $\States$, of size $\statenum$;
    \item the sets of actions for each player $\Actionone$ and $\Actiontwo$, with at most $\actionnum$ actions; and 
    \item the stochastic transition function $\prob \colon \States \times \Actionone \times \Actiontwo \to \Delta(\States)$.
\end{itemize}

\smallskip\noindent{\bf Steps.} 
Given an initial state $\state \in \States$, the game proceeds as follows. 
In each step, both players choose an action simultaneously, $\actionone \in \Actionone$ and $\actiontwo \in \Actiontwo$. 
Based on both actions $(\actionone, \actiontwo)$ and current state $\state$, the next state is drawn
according to the probability distribution $\prob(\state, \actionone, \actiontwo)$. 

\smallskip\noindent{\bf Histories and plays.}
At step $k$ of CSGs, each player possesses information in the form of the finite sequence of the states visited and the actions chosen by both players. A $k$-history $\play^{(k)} = \langle \state_0, \actionone_0, \actiontwo_0, \state_1, \actionone_1, \actiontwo_1, \ldots, \state_k \rangle$ is a finite sequence of states and actions such that, for all steps $0 \le t < k$,
we have $\state_{t+1} \in \support(\prob(\state_t, \actionone_t, \actiontwo_t))$. 
The set of all $k$-histories is denoted by~$\Plays^{(k)}$.
Similarly, a play $\play = \langle \state_0, \actionone_0, \actiontwo_0, \state_1, \actionone_1, \actiontwo_1, \ldots \rangle$ is an infinite sequence of states and actions 
such that, for all steps $t \ge 0$,
we have $\state_{t+1} \in \support(\prob(\state_t, \actionone_t, \actiontwo_t))$. 
The set of all plays is denoted by $\Plays$.
For any state $\state$, the set of all plays starting at $\state$, i.e., $\play = \langle \state_0, \actionone_0, \actiontwo_0, \ldots \rangle$ where $\state_0 = \state$, is denoted by $\Plays_\state$.

\smallskip\noindent{\bf Objectives.}
An objective is a measurable function that assigns a real number to all plays. Player~1 aims to maximize the expectation of the objective, while Player~2 minimizes it. 

\begin{itemize}
    \item {\em Parity objective.} 
    Given a priority function $\prior\colon \States \to \{0, \dots, \parind \}$ with $\parind$ as its index, the {\em parity} objective is an indicator of the even parity condition on minimal priority visited infinitely often in plays. More formally, we define $\parity_\prior \colon \Plays \to \{0, 1\}$ as
    \[
        \parity_\prior(\play) \defas \begin{cases}
            1 & \min \{ \prior(\state) \mid \forall i \ge 0 \; \exists j \ge i \; \state_j = \state \} \text{ is even}\\
            0 & \text{otherwise}
        \end{cases}
    \]

    
    \item {\em Stateful-discounted objective.} Consider $\parind$ discount factors $\discountfactor_1, \cdots, \discountfactor_\parind \in (0, 1]$. Given an assignment function $\assign \colon \States \to [\parind]$, we define the discount function $\discfac \colon \States \to \{ \discountfactor_1, \cdots, \discountfactor_\parind \}$ as $\discfac(\state) \defas \discountfactor_{\assign(\state)}$ for all states $\state \in \States$. Given a reward function $\reward \colon \States \times \Actionone \times \Actiontwo \to [0, 1]$ that assigns a reward value $\reward(\state, \actionone, \actiontwo)$ for all $(\state, \actionone, \actiontwo)$, the stateful-discounted objective $\discounted_\discfac \colon \Plays \to [0, 1]$ is defined as, for all $\play = \langle \state_0, \actionone_0, \actiontwo_0, \cdots \rangle$,
    \[
     \discounted_\discfac(\play) \defas \sum_{i \ge 0} \left (\reward(\state_i, \actionone_i, \actiontwo_i) \discfac(\state_i) \prod_{j < i} 1 - \discfac(\state_j) \right )\,.
    \]
\end{itemize}

\smallskip\noindent{\bf Strategies.}
A strategy is a function that assigns a probability distribution over actions to every finite history and is denoted by $\strategyone \colon \bigcup_{k} \Plays^{(k)} \to \Delta(\Actionone)$ for Player~1 (resp. $\strategytwo \colon \bigcup_{k} \Plays^{(k)} \to \Delta(\Actiontwo)$ for Player~2).
Given strategies $\strategyone$ and $\strategytwo$, the game proceeds as follows.
At step $k$, the current history is some $\play^{(k)} \in \Plays^{(k)}$.
Player~1 (resp. Player~2) chooses an action according to the distribution $\strategyone \left (\play^{(k)} \right )$ (resp. $\strategytwo \left (\play^{(k)} \right )$).
The set of all strategies for Player~1 and Player~2 is denoted by $\Strategyone$ and $\Strategytwo$ respectively.
A {\em stationary} strategy depends on the past
observations only through the current state. 
A stationary strategy for Player~1 (resp. Player~2) is denoted by $\strategyone \colon \States \to \Delta(\Actionone)$ (resp. $\strategytwo \colon \States \to \Delta(\Actiontwo)$). 
The set of all stationary strategies for Player~1 and Player~2 is denoted by $\Strategyone^{S}$ and $\Strategytwo^{S}$ respectively. 
A {\em pure stationary} strategy $\strategyone \colon \States \to \Actionone$ (resp. $\strategytwo \colon \States \to \Actiontwo$) for Player~1 (resp. Player~2) is a stationary strategy that maps to Dirac distributions only. The set of all pure stationary strategies for Player~1 and Player~2 is denoted by $\Strategyone^{PS}$ and $\Strategytwo^{PS}$ respectively.

\smallskip\noindent{\bf Probability space.}
An initial state $\state$ and a pair of strategies ($\strategyone$, $\strategytwo$) induce a unique probability over $\Plays_\state$, endowed with the sigma-algebra generated by the cylinders corresponding to finite histories. We denote by $\PP_\state^{\strategyone, \strategytwo}$ and $\EE_\state^{\strategyone, \strategytwo}$ the probability and the expectation respectively.

We state the determinacy for CSGs with stateful-discounted and parity objectives.

\begin{theorem}[Parity determinacy \cite{martin1998determinacy}]
    For all CSGs, states $\state$, and priority functions $p$,
    \[
        \sup_{\strategyone \in \Strategyone} \inf_{\strategytwo \in \Strategytwo} \EE_{\state}^{\strategyone, \strategytwo}[\parity_p] 
            = \inf_{\strategytwo \in \Strategytwo} \sup_{\strategyone \in \Strategyone}  \EE_{\state}^{\strategyone, \strategytwo}[\parity_p] \,. 
    \]
\end{theorem}

\begin{theorem}[Stateful-discounted determinacy \cite{shapley1953stochastic}]
    For all CSGs, states $\state$, reward functions, and discount functions $\discfac$, we have 
    \[
        \sup_{\strategyone \in \Strategyone^S} \inf_{\strategytwo \in \Strategytwo^S} \EE_{\state}^{\strategyone, \strategytwo}[\discounted_\discfac] =
        \inf_{\strategytwo \in \Strategytwo^S} \sup_{\strategyone \in \Strategyone^S}  \EE_{\state}^{\strategyone, \strategytwo}[\discounted_\discfac] \,. 
    \]
\end{theorem}

\smallskip\noindent{\bf Values.}
The above determinacy results imply that switching the quantification order of strategies do not make 
a difference and leads to the unique notion of value.
The stateful-discounted value for a state $\state$ is defined as 
\[
    \val_\discfac(\state) 
        \defas \sup_{\strategyone \in \Strategyone^S} \inf_{\strategytwo \in \Strategytwo^S} \EE_{\state}^{\strategyone, \strategytwo}[\discounted_\discfac] \,.
\]
We define the parity value $\val_p(\state)$ for a state $\state$ analogously. The limit value for a state $\state$ is defined as
\[
    \val_{\assign}(\state) \defas \lim_{\discountfactor_1 \to 0^+} \cdots \lim_{\discountfactor_\parind \to 0^+} \val_\discfac(\state) \,.
\]

\smallskip\noindent{\bf $\eps$-optimal strategies. }
Given $\eps \geq 0$, a strategy $\strategyone$ for Player~$1$ is \emph{$\eps$-optimal} for the stateful-discounted objective if, for all states $\state \in \States$, we have
\[
    \inf_{\strategytwo \in \Strategytwo^{S}} \EE_\state^{\strategyone, \strategytwo}[ \discounted_\discfac] 
        \ge \val_\discfac(\state) - \eps \,.
\]
We say the strategy is \emph{optimal} if $\eps = 0$.
The notion of $\eps$-optimal strategies for Player $2$ is defined analogously. Similarly, we define $\eps$-optimal strategies for the parity objectives.

\smallskip\noindent{\bf Approximate value problems.}
We consider two value problems stated as follows.

    \begin{tcolorbox}
    \subparagraph{\UndiscVal.} 
     Consider a CSG $G$, a state $\state$, a reward function $\reward$, an assignment function $\assign \colon \States \to [\parind]$, and an additive error $\eps$.
     The transition function $\prob$ and the reward function $\reward$ are represented by rational numbers using at most $\B$ bits. 
     Compute an approximation $v$ of the limit value at state $\state$ such that 
    \[
        |v - \val_{\assign}(\state)| \le \eps\,.
    \]
    \end{tcolorbox}

    \begin{tcolorbox}
    \subparagraph{\ParVal.} 
    Consider a CSG $G$, a state $\state$, a priority function $p$ with index $\parind$, and an additive error $\eps$.
    The transition function $\prob$ is represented by rational numbers using at most $\B$ bits. 
    Compute an approximation $v$ of the parity value at state $\state$ such that 
    \[
        |v - \val_p(\state)| \le \eps\,.
    \]
    \end{tcolorbox}

\section{Overview of Results}
\label{Section: Previous and our results}
We first discuss the previous results in the literature, and then, we show our contributions.

\smallskip\noindent{\bf Previous results.}
We discuss the previous works on computing the approximation of limit and parity values in CSGs. A natural approach for these computational problems is via the theory of reals. We first recall the main computational result of the theory of reals, which is a specialization of Theorem 1 of \cite{basu1999new}.

\begin{theorem}[\cite{basu1999new}]
    \label{Result: Existential theory of reals}
    Consider $\ell$ variables $x_1, \cdots, x_\ell$ and the set of polynomials $\setPolynomials = \{P_1, \cdots, P_k \}$, where, for all $i \in [k]$, we have  $P_i$ is a polynomial in $x_1, \cdots, x_\ell$ of degree at most $\degree$ with integer coefficients of bit-size at most $\B$. 
    Let $X_1, \cdots, X_\parind$ be a partition of $x_1, \cdots, x_\ell$ into $\parind$ subsets such that $X_i$ has size $\ell_i$. Let
    \[
        \Phi = (Q_\parind X_\parind) \; \cdots \; (Q_1 X_1) \quad \phi(P_1, \cdots, P_k)
    \]
    be a sentence with $\parind$ alternating quantifiers $Q_i \in \{ \exists, \forall \}$ such that $Q_{i+1} \neq Q_i$, and $\phi(P_1, \cdots, P_k)$ is a quantifier-free formula with atomic formulas of the form $P_i \bowtie 0$ where $\bowtie \, \in \{<, >, = \}$. Then, there exists an algorithm to decide the truth of $\Phi$ in time 
    \[
        k^{\prod_i \calO(\ell_i+1)} \cdot D^{\prod_i \calO(\ell_i)} \cdot \calO(\len(\phi) \B^2)\,,
    \]
    where $\len(\phi)$ is the length of the quantifier-free formula $\phi$.
\end{theorem}
Along with the above algorithmic result, the following complexity result also follows from~\cite{basu1999new}: if there is constant number of quantifier alternations, then complexity is PSPACE, and in general the complexity is EXPSPACE.
We now discuss the algorithms and complexity results from the literature for the limit value of stateful discounted-sum objectives. 
The basic computational approach is via the theory of reals. 
For a single discount factor, the reduction to the theory of reals and dealing with its limit (which corresponds to limit-average objectives) was presented in~\cite{chatterjee2008stochastic}. 
In the general case ($\parind$ discount factors), each limit can be considered as the quantification $\exists \eps_{i'}\; \forall \eps_{i} \leq \eps_{i'} $ in the theory of reals. 
Thus, concurrent stochastic games with the limit value of stateful-discounted objectives can be reduced to the theory of reals with quantifier alternation.
This reduction gives a theory of reals sentence with the following parameters: 
\[
    \ell = \calO(\actionnum^2\statenum),
        \quad k = \calO(\actionnum^2\statenum),
        \quad \degree = 4,
        \quad \prod_i (\ell_i+1) = \calO(2^\parind \actionnum^2\statenum)\,,
\]
Applying \Cref{Result: Existential theory of reals} to the reduction we obtain the following result.
\begin{theorem}[\UndiscVal: Previous Result]
\label{Result: Previous computational result for limit value}
    For the \UndiscVal\ problem, the following assertions hold.
    \begin{enumerate}
        \item The problem is in \textnormal{EXPSPACE}; and
        \item the problem can be solved in time $\exp \left( \calO \left( 2^\parind \actionnum^2\statenum + \log(1/\eps) +   \log (\B) \right) \right)$.
    \end{enumerate}
\end{theorem}
For parity objectives, the result of~\cite{gimbert2005discounting,de2003discounting} reduces CSGs with parity objectives to CSGs with the limit value of stateful-discounted objectives. The reduction is achieved as follows. Consider the formula $R(a_0,a_1, \ldots, a_{2n-1})$ from~\cite{de2003discounting}, which is a formula with multiple discount factors. Since for stateful-discounted objectives the mapping is contractive, the fixpoints are unique (least and greatest fixpoints coincide). The last sentence of \cite[Theorem~4]{de2003discounting} states that the limit of $R(a_0, a_1, \ldots, a_{2n-1})$ corresponds to the value for parity objectives. The Pre operator of the formula corresponds to the Bellman-operator for stateful-discounted objectives, which establish the connection to stateful-discounted games. This connection is made more explicit in the construction provided in \cite[Section 2.2]{gimbert2005discounting}. This linear reduction and the above theorem lead to similar results for parity objectives. 
Besides this reduction to the theory of reals, there are two other approaches for the \ParVal\ problem.
First, we can consider the nested fixpoint characterization as provided in~\cite{de2001quantitative} and a reduction to the theory of reals with quantifier alternation.
However, this does not lead to a better complexity. 
Second, a different approach is presented in~\cite[Chapter 8]{chatterjee2006complexity,chatterjee2007stochastic}.
This approach has the following components: 
(a)~it enumerates over all possible subsets of actions for every state;
(b)~for each of the enumeration, it requires a solution of a qualitative value problem (or limit-sure winning) in concurrent stochastic games with parity objectives, and the value-approximation for concurrent stochastic games with reachability objectives.  
This approach gives PSPACE complexity and the algorithmic complexity is $\exp \left( \calO \left(\actionnum\statenum + \parind \log(\statenum) + \log(\log(1/\eps)) + \log (\B) \right) \right)$.

\smallskip\noindent{\bf Our contributions.} 
Our main results are as follows:

\begin{theorem}[\UndiscVal: Complexity and Algorithm]
\label{Result: Computational result for limit value}
    For the \UndiscVal\ problem, the following assertions hold.
    \begin{enumerate}
        \item 
        \label{Item: Complexity class for limit value} 
        The problem is in \textnormal{TFNP[NP]}; and

        \item 
        \label{Item: Algorithmic time for limit value}
        the problem can be solved in time $\exp \left( \calO \left( \statenum\parind \log(\actionnum) + \log(\B) + \log(\log(1/\eps)) \right) \right)$.
    \end{enumerate}
\end{theorem}

\begin{theorem}[\ParVal: Complexity and Algorithm]
\label{Result: Computational result for parity value}
    For the \ParVal\ problem, the following assertions hold.
    \begin{enumerate}
        \item
        \label{Item: Complexity class for parity value} 
        The problem is in \textnormal{TFNP[NP]}; and
        \item 
        \label{Item: Algorithmic time for parity value}
        the problem can be solved in time $\exp \left( \calO \left( \statenum\parind \log(\actionnum) + \log(\B) + \log(\log(1/\eps)) \right) \right)$.
    \end{enumerate}
\end{theorem}

%% file: mathematical_properties.tex
\section{Mathematical Properties}

In this section, we present a new approach to the limit value approximation via the stateful-discounted value (\Cref{Result: Approximation of undiscounted value by discounted value}). 
We follow the approach of \cite{attia2019formula} extending it step by step using similar arguments.
We use this technical result to improve complexities and algorithmic bounds of computing $\eps$-approximation of the limit and parity values. 
The section is organized as follows. In \Cref{Section: Mathematical properties notations}, we present some useful definitions and previous results in the literature.
In \Cref{Section: Bounds on roots of polynomials}, we present a bound on the roots of multi-variate polynomials which is used to establish a connection between the stateful-discounted and limit values. 
In \Cref{Section: Characterization of discounted value,Section: Characterization of limit value}, we introduce new characterizations of the stateful-discounted and limit values. 
In \Cref{Section: Approximation of limit value}, we finally show \Cref{Result: Approximation of undiscounted value by discounted value}. 

\subsection{Notations and Selected Results from Literature}
\label{Section: Mathematical properties notations}
We first present some basic notations and definitions. 

\smallskip\noindent{\bf Basic notations.} 
Given a positive integer $k$, we define $\bit(k) \defas \lceil \log_2(k+1) \rceil$. 
For a rational $k_1/k_2$, we define $\bit(k_1/k_2) \defas \bit(k_1) + \bit(k_2)$. 
Given a real $x$, the sign function is
\[
    \sign(x) =
    \begin{cases}
    -1 & \text{if } x < 0, \\
    0 & \text{if } x = 0, \\
    1 & \text{if } x > 0.
    \end{cases}
\]
Moreover, we use the classic arithmetic with infinity, i.e., $x + \infty = \infty$ and $x - \infty = -\infty$.

\smallskip\noindent{\bf Polynomials.} 
A {\em uni-variate} polynomial $P$ of degree $\degree$ with integer coefficients of bit-size $\B$ is defined as $P(x) \defas \sum_{i=0}^{\degree} c_ix^i$ where $|c_i| < 2^\B$. We define $\lVert P \rVert_\infty \defas \max(|c_0|, \cdots, |c_\degree|)$. 
A $k$-variate polynomial $P$ in $x_1, \cdots, x_k$ of degree $\degree_1, \cdots, \degree_k$ with integer coefficients of bit-size $\B$ is defined as 
\[
    P(x_1, \cdots, x_k) \defas \sum_{0 \le i_1 \le \degree_1} \cdots \sum_{0 \le i_k \le \degree_k} c_{i_1, \cdots, i_k} \prod_{j=1}^k x_j^{i_j}\,,
\]
where $|c_{i_1, \cdots, i_k}| \le 2^\B$. 
Polynomial $P$ is nonzero if $c_{i_1, \cdots, i_k} \neq 0$ for some $i_1, \cdots, i_k$.
We say $\alpha$ is a root of $P$ if $P(\alpha) = 0$. 
In this work, we only consider real roots. 

\smallskip\noindent{\bf Matrix notations.} 
Given a square matrix $M$, we denote the determinant of $M$ by $\det(M)$ and denote the signed sum of all minors of $M$ by $S(M)$. 
Given two $k \times \ell$ matrices $M_1$ and $M_2$, we denote the Hadamard product of $M_1$ and $M_2$ by $M_1 \odot M_2$. 
Given a positive integer $k$, often implicitly clear from context, we denote by $\mathbf{1}$ (resp. $\mathbf{0}$) the $k$-dimensional vector with all elements equal to $1$ (resp. $0$) and denote by $\id$ the $k \times k$ identity matrix.

\smallskip\noindent{\bf Matrix games.} 
A matrix $M$ defines a game played between two opponents, where rows (resp. columns) correspond to possible actions for the row- (resp. \mbox{column-)} player, and the entry $(M)_{i, j}$ is the reward the column-player pays the row-player when the pair of actions $(i, j)$ is chosen.
The value of a matrix game, denoted $\val M$, is the maximum amount the row-player can guarantee, i.e., the amount they can obtain regardless of the column-player's strategy. 

Bellow, we recall some useful results from the literature regarding the value of a matrix game (\Cref{Result: Value of matrix games}) and the determinant of a polynomial matrix (\Cref{Result: Det of poly matrix}). 

\begin{lemma}[\protect{\cite[Thm. 2]{shapley1952basic}}]
\label{Result: Value of matrix games}
    Consider a matrix game $M$. Then, there exists a square sub-matrix $M_0$ such that 
        $S(M_0) \neq 0$ and 
        $\val(M) = \frac{\det(M_0)}{S(M_0)}$.
\end{lemma}

\begin{lemma}[\protect{\cite[Prop. 8.12]{basu2006algorithms}}]
\label{Result: Det of poly matrix}
Consider a $k \times k$ matrix $M$ whose entries are polynomials in $x_1, \cdots, x_\ell$ of degrees $\degree_1, \cdots, \degree_\ell$ with integer coefficients of bit-size $\B$. Then, $\det(M)$ is a polynomial in $x_1, \cdots, x_\ell$ of degrees $k\degree_1, \cdots, k\degree_\ell$ with integer coefficients of bit-size $k\B + k \bit(k) + \ell \bit(k \max(\degree_1, \cdots, \degree_\ell) + 1)$.
\end{lemma}

\begin{corollary}
\label{Result: S of poly matrix}
Consider a $k \times k$ matrix $M$ whose entries are polynomials in $x_1, \cdots, x_\ell$ of degrees $\degree_1, \cdots, \degree_\ell$ with integer coefficients of bit-size $\B$. Then, $S(M)$ is a polynomial in $x_1, \cdots, x_\ell$ of degrees $k\degree_1, \cdots, k\degree_\ell$ with integer coefficients.
\end{corollary}
\begin{proof}
    By definition, $S(M)$ is the signed sum of all minors of $M$. Therefore, the result follows from \Cref{Result: Det of poly matrix}.
\end{proof}

\subsection{Bounds on Roots of Polynomials with Integer Coefficients}
\label{Section: Bounds on roots of polynomials}
In this subsection, we present a bound on the roots of multi-variate polynomials $P$ with integer coefficients (\Cref{Result: n-variate root separation from 0}). This result shows that there exists a region close to~$\mathbf{0}$ within which $P$ does not have a root. We use this technical result to establish a connection between the stateful-discounted value and limit value.

\begin{lemma}
\label{Result: n-variate root separation from 0}
    Consider a nonzero polynomial $P$ in $x_1, \cdots, x_\ell$ of degrees $\degree_1, \cdots, \degree_\ell$ with integer coefficients of bit-size $\B$. 
    Let $\degree \defas \max(\degree_1, \cdots, \degree_\ell)$ and $\B_1 \defas 4\ell\bit(\degree) + \B + 1$. Then,
    \begin{gather*}
        \forall \, x_1 \in \left (0, \exp(-\B_1) \right ] \quad \forall \, x_2 \in \left (0, (x_1)^{\degree + 1} \right ] \quad \cdots \quad \forall \, x_\ell \in \left (0, (x_{\ell-1})^{\degree + 1} \right ]\\
        |P(x_1, \cdots, x_\ell)| \ge \exp(\B_1 - \ell) \cdot (x_\ell)^{\degree+1}\,.
    \end{gather*}
    
\end{lemma}

\begin{proof}
    We proceed with the proof by induction on $\ell$.
    
    \smallskip\noindent{\em Base case $\ell=1$.} 
    Consider $P(x_1) = c_0 + c_1x_1 + \cdots + c_{\degree_1}x_1^{\degree_1}$. Note that $P$ is nonzero. 
    Let $k$ be the smallest index where $c_k \neq 0$. Therefore, for all $x_1 \le \exp(-\B_1)$, we have
    \begin{align*}
        |P(x_1)| &\ge |c_k|x_1^k - \sum_{j > k} |c_j|x_1^j & (\text{triangle inequality})\\
        &\ge |c_k|x_1^k - x_1^{k+1} \sum_{j > k} |c_j| & (x_1^j \le x_1^{k+1})\\
        &\ge x_1^k - x_1^{k+1} \sum_{j > k} |c_j| & (|c_k| \ge 1)\\
        &\ge x_1^k - \degree\exp(\B)x_1^{k+1} & (|c_j| \le \exp(\B))\\
        &\ge x_1^k \left (1 - \degree\exp(\B)x_1 \right ) & (\text{rearrange})\\
        &\ge x_1^k \left (\exp(\B_1)x_1 - \degree\exp(\B)x_1 \right ) & (x_1 \le \exp(-B_1))\\
        &\ge \exp(\B_1 - 1)x_1^{\degree+1}\,, & (k \le \degree \text{ and } \degree\exp(\B) \le \exp(\B_1 - 1))\\
    \end{align*}
    which completes the case.

    \smallskip\noindent{\em Induction case $\ell > 1$.}
    We partition $P$ into 
    \[
        P(x) = P_0 + P_1(x_1) + \cdots + P_\ell(x_1, \cdots, x_\ell)\,.
    \]
    Note that $P$ is nonzero. 
    Therefore, let $i$ be the smallest index where $P_i \neq 0$. By the choice of $i$, for all $j < i$, we have that $P_j = 0$. 
    Fix $x_1, \cdots, x_\ell$ such that $x_1 \le \exp(-\B_1)$, and $x_i \le x_{i-1}^{\degree+1}$ for $i > 1$. 
    We show that 
    \begin{equation}
    \label{Eq: P is not 0}
        \left \lvert P_i(x_1, \cdots, x_i) \right \rvert - \sum_{j > i} \left \lvert P_j(x_1, \cdots, x_j) \right \rvert \ge \exp(\B_1 - \ell) x_\ell^{\degree+1}\,,
    \end{equation}
    which concludes the case.
    
    To bound the term $\left \lvert P_i(x_1, \cdots, x_i) \right \rvert$ of LHS of \Cref{Eq: P is not 0} from below, we define 
    \[
        Q(y) \defas P_i(x_1 \cdots, x_{i-1}, y) = \sum_{j = 1}^{\degree_i} C_j(x_1, \cdots, x_{i-1}) y^j\,,
    \]
    where for all $1 \le j \le \degree_i$, we have
    \[
        C_j(x_1, \cdots, x_{i-1}) \defas  \sum_{j_1=0}^{\degree_1} \cdots \sum_{j_{i-1}=0}^{\degree_{i-1}} c_{j_1, \cdots, j_{i-1}, j, 0, \cdots, 0} \prod_{k=1}^{i-1} x_k^{j_k}\,.
    \]
    Note that $Q$ is a univariate polynomial of degree at most $\degree_i$, where the absolute value of coefficients are at most 
    \begin{align*}
        |C_j(x_1, \cdots, x_{i-1})| &\le \exp(\B) (\degree+1)^{i-1}\\
        &\le \exp(\B + (i-1) \bit(\degree+1))\\ 
        &\le \exp(\B + 2\ell\bit(\degree))\\
        &\le \exp(\B_1 - \bit(\degree) - (i+1))\,
    \stepcounter{equation}\tag{\theequation}\label{Eq: Upper bound of coeffs of G}
    \end{align*}
    where in the first inequality we use $c_{j_1, \cdots, j_{i-1}, j, 0, \cdots, 0} \le \exp(\B)$, in the second inequality we use $(\degree+1)^{i-1} \le \exp((i-1) \bit(\degree+1))$, in the third inequality we use $i-1 \le \ell$ and $\bit(\degree+1) \le 2\bit(\degree)$, and in the fourth inequality we use $\B + 2\ell\bit(\degree) + \bit(\degree) + (i+1) \le \B_1$. For all $1 \le j \le \degree_i$, note that $C_j$ is a polynomial in $x_1, \cdots, x_{i-1}$ of degrees $\degree_1, \cdots, \degree_{i-1}$ with integer coefficients of bit-size $\B$. Therefore, if $C_j$ is nonzero, then by induction, we have
    \begin{equation}
    \label{Eq: Lower bound of coeffs of G}
        C_j(x_1, \cdots, x_{i-1}) \ge \exp(\B_1 - (i-1)) x_{i-1}^{\degree+1} \ge \exp(\B_1 - (i-1)) x_{i}\,.
    \end{equation}
    Consider $Q(y) = c_0 + c_1y + \cdots + c_{d_i}y^{\degree_i}$, where for all $j$, we have $c_j = C_j(x_1, \cdots, x_{i-1})$. 
    Since $P_i$ is nonzero, by the definition of $Q$, there exists $j$ such that $C_j \neq 0$. 
    Therefore, by \Cref{Eq: Lower bound of coeffs of G}, we have that $Q$ is also nonzero. 
    Let $k$ be the smallest index where $c_k \neq 0$. Therefore,
    \begin{align*}
        &|P_i(x_1, \cdots, x_i)|\\
        &\quad = |Q(x_1)| & (\text{def. } Q)\\
        &\quad \ge |c_k x_i^k| - \sum_{j > k} |c_j x_i^j| & (\text{triangle inequality})\\
        &\quad \ge |c_k x_i^k| - x_i^{k+1} \sum_{j > k} |c_j| & \left (x_i^j \le x_i^{k+1} \right )\\
        &\quad = x_i^k \left (|c_k| - x_i \sum_{j > k} |c_j| \right ) & (\text{rearrange}) \\
        &\quad \ge x_i^k \left (|c_k| - \exp(\B_1 - (i+1)) x_i  \right ) & (\text{\Cref{Eq: Upper bound of coeffs of G}})\\ 
        &\quad \ge x_i^k \left (\exp(\B_1 - (i-1))x_i - \exp(\B_1 - (i+1)) x_i  \right ) & (\text{\Cref{Eq: Lower bound of coeffs of G}}) \\
        &\quad \ge x_i^{\degree+1} \left (\exp(\B_1 - (i-1)) - \exp(\B_1 - (i+1))  \right )\,. & (k \le \degree)
        \stepcounter{equation}\tag{\theequation}\label{Eq: LHS of Sign(P_i) = Sign(P)}
    \end{align*}
    To bound the term $\sum_{j > i} \left \lvert P_j(x_1, \cdots, x_j) \right \rvert$ of the LHS of \Cref{Eq: P is not 0} from above, consider $P_j$ for all $j > i$. Recall that $P_j$ is a polynomial in $x_1, \cdots, x_j$ of degrees at most $\degree_1, \cdots, \degree_j$ with integer coefficients of bit-size $\B$. Therefore, the number of terms in $P_j$ is at most $(\degree + 1)^j$, and for each term, the degree of $x_j$ is at least 1. Therefore, 
    \begin{align*}
        \left \lvert P_j(x_1, \cdots, x_j) \right \rvert &\le \left (\degree + 1 \right )^j \cdot \exp(\B) \cdot x_j & (c_{i_1, \cdots, i_j, 0, \cdots, 0} \le \exp(\B))\\
        &\le \left (\degree + 1 \right )^\ell \cdot \exp(\B)  \cdot x_j & \left ((\degree+1)^j \le (\degree+1)^\ell \right )\,.
    \end{align*}
    Hence,
    \begin{align*}
        \sum_{j > i} \left \lvert P_j(x_1, \cdots, x_j) \right \rvert &\le (\degree + 1)^\ell \exp(\B) \sum_{j > i} x_j \\
        &\le \ell (\degree + 1)^\ell \exp(\B) \cdot x_{i+1} \\
        &\le \exp\left (\bit(\ell) + \ell\bit(\degree+1) + \B \right ) \cdot x_{i+1} \\
        &\le \exp\left (\bit(\ell) + 2\ell\bit(\degree) + \B \right ) \cdot x_{i+1} \\
        &\le \exp(\B_1 - (i+1)) x_{i+1}\\
        &\le \exp(\B_1 - (i+1)) x_{i}^{\degree+1}\,,
        \stepcounter{equation}\tag{\theequation}\label{Eq: RHS of Sign(P_i) = Sign(P)}
    \end{align*}
    where in the second inequality we use $x_j \le x_{i+1}$ for all $j \ge i+1$, in the third inequality we use $\ell \cdot (\degree + 1)^\ell \le \exp(\bit(\ell) + \ell\bit(\degree+1))$, in the fourth inequality we use $\bit(\degree+1) \le 2\bit(\degree)$, in the fifth inequality we use $\bit(\ell) + 2\ell\bit(\degree) + \B + (i+1) \le \B_1$, and in the sixth inequality we use $x_{i+1} \le x_i^{\degree+1}$.
    By combining \Cref{Eq: LHS of Sign(P_i) = Sign(P),Eq: RHS of Sign(P_i) = Sign(P)}, we get
    \begin{align*}
        \left \lvert P_i(x_1, \cdots, x_i) \right \rvert - \sum_{j > i} \left \lvert P_j(x_1, \cdots, x_j) \right \rvert &\ge \exp(\B_1 - i)x_i^{\degree+1}\\
        &\ge \exp(\B_1 - \ell) x_i^{\degree+1} & (i \le \ell)\\
        &\ge \exp(\B_1 - \ell) x_\ell^{\degree+1}\,, & (x_\ell \le x_i)\\
    \end{align*}
    which concludes the claim, completes the case, and yields the result.
\end{proof}

We notice that \Cref{Result: n-variate root separation from 0} may be compared with the minimal distance of two semi-algebraic sets that do not intersect, formalized for example in~\cite[Corollary 3.8]{schaefer2017FixedPointsNash}, but these results are incomparable.
For example, consider a multivariate polynomial $P$ and the algebraic sets $\{ x \in \RR^\ell : P(x_1, \ldots, x_\ell) = 0 \}$ and $\{ 0 \}$. 
These sets intersect at the origin if $P(0, \ldots, 0) = 0$.
In this case, there is no minimal distance between these two sets, while \Cref{Result: n-variate root separation from 0} provides a region for them.
To the best of our knowledge, \Cref{Result: n-variate root separation from 0} is a new result.

\subsection{Characterization of Stateful-discounted Value}
\label{Section: Characterization of discounted value}
In this subsection, we introduce a new characterization of the stateful-discounted value in CSGs (\Cref{Result: Characterization of discounted value,Result: Properties of W}), which generalizes the result presented in \cite{attia2019formula} from a single discount factor to multiple discount factors. In particular, \Cref{Result: Characterization of discounted value} generalizes Theorem 1 of \cite{attia2019formula}.

\smallskip\noindent{\bf Stateful-discounted payoff.} Consider a CSG $G$, a state $\state$, a reward function $\reward$, and a discount function $\discfac$. Given a pair of stationary strategies $(\strategyone, \strategytwo)$, we define the stateful-discounted payoff as 
\[
    \payoff^{\strategyone, \strategytwo}(\state) \defas \EE_\state^{\strategyone, \strategytwo} [\discounted_\discfac] \,.
\]
By fixing $\strategyone$ and $\strategytwo$, we obtain a transition function
\[
    \prob^{\strategyone, \strategytwo}(\state, \state') \defas \sum_{\substack{\actionone \in \Actionone \\ \actiontwo \in \Actiontwo}} \strategyone(\state)(\actionone) \cdot \strategytwo(\state)(\actiontwo) \cdot \prob(\state, \actionone, \actiontwo)(\state')  \,,
\]
which is described as a matrix, i.e., $\prob^{\strategyone, \strategytwo} \in \RR^{\statenum \times \statenum}$. The stage reward function is defined as 
\[
    \reward^{\strategyone, \strategytwo}(\state) \defas \sum_{\substack{\actionone \in \Actionone \\ \actiontwo \in \Actiontwo}} \strategyone(\state)(\actionone) \cdot \strategytwo(\state)(\actiontwo) \cdot \reward(\state, \actionone, \actiontwo) \,,
\]
which is described as a vector, i.e., $\reward^{\strategyone, \strategytwo} \in \RR^{\statenum}$. Therefore, the Bellman operator defined in~\cite{shapley1953stochastic} can be written as a recursive expression: 
\[
    \payoff^{\strategyone, \strategytwo} = \discfac \odot \reward^{\strategyone, \strategytwo} + (\mathbf{1} - \discfac) \odot \left (\prob^{\strategyone, \strategytwo} \payoff^{\strategyone, \strategytwo} \right ) \,.
\]
The matrix $\id - \left ( \, (\mathbf{1} - \discfac) \, \mathbf{1}^\top \, \right ) \odot \prob^{\strategyone, \strategytwo}$ is strictly diagonally dominant, and therefore, is invertible. By Cramer's rule, we have 
\begin{equation}
\label{Eq: Discounted payoff}
    \payoff^{\strategyone, \strategytwo}(\state) = \frac{\deter_\discfac^{\state}(\strategyone, \strategytwo)}{\deter_\discfac(\strategyone, \strategytwo)} \,,
\end{equation}
where $\deter_\discfac(\strategyone, \strategytwo) \defas \det \left ( \id - \left ( \, (\mathbf{1} - \discfac) \, \mathbf{1}^\top \, \right ) \odot \prob^{\strategyone, \strategytwo} \right )$ and $\deter_\discfac^{\state}(\strategyone, \strategytwo)$ is the determinant of an $\statenum \times \statenum$ matrix derived by substituting the $\state$-th column of the matrix $\id - \left ( (\mathbf{1} - \discfac)\mathbf{1}^\top \right ) \odot \prob^{\strategyone, \strategytwo}$ with $\discfac \odot \reward^{\strategyone, \strategytwo}$.

\smallskip\noindent{\bf Auxiliary matrix game $W_\discfac^\state(z)$.} We define a matrix game where the actions of each player are the pure stationary strategies in the stochastic game. The payoff of the game is obtained by the linearization of the quotient in \Cref{Eq: Discounted payoff}. More formally,
for all parameters $z \in \RR$, $\widehat{\strategyone} \in \Strategyone^{PS}$, and $\widehat{\strategytwo} \in \Strategytwo^{PS}$, we define the payoff of the matrix game as
\[
    W_\discfac^\state(z)[\widehat{\strategyone}, \widehat{\strategytwo}] \defas \deter_\discfac^\state(\widehat{\strategyone}, \widehat{\strategytwo}) - z \cdot \deter_\discfac(\widehat{\strategyone}, \widehat{\strategytwo}) \, .
\]
The value of $W_\discfac^\state(z)$ is denoted by $\val \left (W_\discfac^\state(z) \right )$.

The following statements (\Cref{Result: Properties of W} and \Cref{Result: Characterization of discounted value}) connect the stateful-discounted value with the value of the matrix game.

\begin{lemma}
\label{Result: Properties of W}
    Consider a CSG $G$, a state $\state$, a reward function, and an assignment function $\assign \colon \States \to [\parind]$. Then, the following assertions hold.
    \begin{enumerate}
        \item
        \label{Item: Continuity of W}
        The map $(z, \discountfactor_1, \cdots, \discountfactor_\parind) \mapsto \val(W_\discfac^\state(z))$ is continuous;

        \item
        \label{Item: Monotonicity of W}
        for all discount factors $\discountfactor_1, \cdots, \discountfactor_\parind$ and $z_1, z_2 \in \RR$ such that $z_1 \le z_2$, we have that
        $\val \left (W_\discfac^\state(z_1) \right ) \ge \val \left (W_\discfac^\state(z_2) \right ) + (z_2 - z_1) \left( \min_i \discountfactor_i \right)^n$, in particular, $z \mapsto \val(W_\discfac^\state(z))$ is strictly decreasing; and        
        \item 
        \label{Item: W is 0 at discounted value}
        for all discount factors $\discountfactor_1, \cdots, \discountfactor_\parind$, we have $\val \left (W_\discfac^\state(\val_\discfac(\state)) \right ) = 0$.
    \end{enumerate}
\end{lemma}

\begin{corollary}
\label{Result: Characterization of discounted value}
    Consider a CSG $G$, a state $\state$, a reward function, and a discount function~$\discfac$. Then, $\val_\discfac(\state)$ is the unique \(z^* \in \RR\) such that  
    \[
        \val \left (W_\discfac^\state(z^*) \right ) = 0 \,.
    \]    
\end{corollary}
Below we define randomized strategies in the matrix game derived from stationary strategies in concurrent stochastic games.

\smallskip\noindent{\bf Strategies for the matrix game.}
Given a Player-1 (resp. Player~2) stationary strategy $\strategyone$ (resp. $\strategytwo$), we denote by $\bm{\strategyone}$ (resp. $\bm{\strategytwo}$) a randomized strategy for the matrix game $W_\discfac^\state(z)$ defined as
\[
    \bm{\strategyone}(\widehat{\strategyone}) \defas \prod_{\state \in \States} \strategyone(\state)(\widehat{\strategyone}(\state)) \quad \forall \widehat{\strategyone} \in \Strategyone^{PS} \,.
\]
Also, we define
\[
    W_\discfac^\state(z)[\bm{\strategyone}, \bm{\strategytwo}] \defas \sum_{\widehat{\strategyone} \in \Strategyone^{PS}} \sum_{\widehat{\strategytwo} \in \Strategytwo^{PS}} \bm{\strategyone}(\widehat{\strategyone}) \cdot \bm{\strategytwo}(\widehat{\strategytwo}) \cdot W_\discfac^\state(z)[\widehat{\strategyone}, \widehat{\strategytwo}] \,. 
\]
The following result is instrumental to prove that the value of the matrix game is strictly decreasing (\Cref{Result: Properties of W}-\Cref{Item: Monotonicity of W}).
\begin{lemma}
\label{Result: Determinant of diagonally dominating matrix}
    Consider a $k \times k$ stochastic matrix $M$ and a discount function $\discfac \colon [k] \to \{ \discountfactor_1, \cdots, \discountfactor_\parind \}$. Then, we have
    \[
        \det \left (\id - \left ( (\mathbf{1} - \discfac)\mathbf{1}^\top \right ) \odot M \right ) \ge \left (\min_{i} \discountfactor_i \right )^k \, .
    \]
\end{lemma}

\begin{proof}
    Fix $\widehat{M} \defas \id - \left ( (\mathbf{1} - \discfac)\mathbf{1}^\top \right ) \odot M$. 
    We claim that the matrix $\widehat{M}$ is a strictly diagonally dominant matrix. Indeed, $M$ is a stochastic matrix. Therefore, for the $i$-th row of $\widehat{M}$, we have 
    \begin{align*}
        \widehat{M}_{i,i} - \sum_{j \neq i} |\widehat{M}_{i,j}| &= 1 - \left (1 - \discfac(i) \right ) \sum_{j} M_{i,j} & (\text{def. } \widehat{M})\\
        &= \discfac(i) \,. & (M \text{ is a stochastic matrix})
    \end{align*}
    Consider the (possibly complex) eigenvalues $\xi_1, \cdots, \xi_k$. By Gershgorin circle theorem \cite{gershgorin1931uber}, for all $i$, we have that $| \xi_i - 1 | \le 1 - \min_i \discountfactor_i$. 
    Therefore, we have
    \[
        \det(\widehat{M}) 
            = \prod_i \xi_i
            \ge  \left( \min_i \discountfactor_i \right) ^k \,,
    \]
    which concludes the proof.
\end{proof}

The following result is instrumental to prove that the stateful-discounted value forces the matrix game to have value zero (\Cref{Result: Properties of W}-\Cref{Item: W is 0 at discounted value}).
\begin{lemma}
\label{Result: generalization of W}
    Consider a CSG $G$, a state $\state$, a reward function, a discount function $\discfac$, and a Player-1 stationary strategy $\strategyone$. Then, for all $\widehat{\strategytwo} \in \Strategytwo^{PS}$ and $z \in \RR$, we have
    \[
        W_\discfac^\state(z)[\bm{\strategyone}, \widehat{\strategytwo}] = \deter_\discfac^\state(\strategyone, \widehat{\strategytwo}) - z \cdot \deter_\discfac(\strategyone, \widehat{\strategytwo}) \,,
    \]
    where $\bm{\strategyone}$ is the strategy for $W_\discfac^\state(z)$ derived by the strategy $\strategyone$.
\end{lemma}
\begin{proof}
    For all $\strategyone \in \Strategyone^{S}$ and $\strategytwo \in \Strategytwo^{S}$, we define $M(\strategyone, \strategytwo) \defas \id - \left ( (1 - \discfac)1^\top \right ) \odot \prob^{\strategyone, \strategytwo}$. We also denote by $\strategyone_{\state \to \actionone}$ the stationary strategy derived from $\strategyone$ where Player~1 chooses the action $\actionone$ at the state $\state$. By definition, $\deter_\discfac(\strategyone, \strategytwo) = \det(M(\strategyone, \strategytwo))$.
    We now show that 
    \[
        \deter_\discfac(\strategyone, \widehat{\strategytwo}) = \sum_{\widehat{\strategyone} \in \Strategyone^{PS}} \bm{\strategyone}(\widehat{\strategyone}) \cdot \deter_\discfac(\widehat{\strategyone}, \widehat{\strategytwo}) \,.
    \]
    Indeed, the $\state$-th row of $M(\strategyone, \widehat{\strategytwo})$ only depends on $\strategyone(\state)$. Therefore, by multi-linearity of the determinant, we have
    \begin{align*}
        \det(M(\strategyone, \widehat{\strategytwo})) &= \det \left (\sum_{\actionone \in \Actionone} \strategyone(\state)(\actionone) M(\strategyone_{\state \to \actionone}, \widehat{\strategytwo}) \right )\\
        &= \sum_{\actionone \in \Actionone} \strategyone(\state)(\actionone) \cdot \det(M(\strategyone_{\state \to \actionone}, \widehat{\strategytwo})) \,.
    \end{align*}
    Hence, by induction, we have
    \[
    \deter_\discfac(\strategyone, \widehat{\strategytwo}) = \det(M(\strategyone, \widehat{\strategytwo})) = \sum_{\widehat{\strategyone} \in \Strategyone^{PS}} \bm{\strategyone}(\widehat{\strategyone}) \cdot \det(M(\widehat{\strategyone}, \widehat{\strategytwo})) = \sum_{\widehat{\strategyone} \in \Strategyone^{PS}} \bm{\strategyone}(\widehat{\strategyone}) \cdot \deter_\discfac(\widehat{\strategyone}, \widehat{\strategytwo}) \,.
    \]
    By similar arguments, we can show that
    \[
    \deter_\discfac^\state(\strategyone, \widehat{\strategytwo}) = \sum_{\widehat{\strategyone} \in \Strategyone^{PS}} \bm{\strategyone}(\widehat{\strategyone}) \cdot \deter_\discfac^\state(\widehat{\strategyone}, \widehat{\strategytwo}) \,.
    \]
    Therefore, 
    \begin{align*}
        W_\discfac^\state(z)[\bm{\strategyone}, \widehat{\strategytwo}] &= \sum_{\widehat{\strategyone} \in \Strategyone^{PS}} \bm{\strategyone}(\widehat{\strategyone}) \cdot W_\discfac^\state(z)[\widehat{\strategyone}, \widehat{\strategytwo}]\\
        &= \sum_{\widehat{\strategyone} \in \Strategyone^{PS}} \bm{\strategyone}(\widehat{\strategyone}) \left [ \deter_\discfac^\state(\widehat{\strategyone}, \widehat{\strategytwo}) - z \cdot \deter_\discfac(\widehat{\strategyone}, \widehat{\strategytwo}) \right ]\\
        &= \deter_\discfac^\state(\strategyone, \widehat{\strategytwo}) - z \cdot \deter_\discfac(\strategyone, \widehat{\strategytwo}) \,,
    \end{align*}
    which completes the proof.
\end{proof}

\begin{proof}[Proof of \Cref{Result: Properties of W}]
    We prove the three items as follows.
    \begin{enumerate}
        \item 
        By the continuity of the determinant, the entries of $W_\discfac^\state(z)$ depend continuously on parameters $z, \discountfactor_1, \cdots, \discountfactor_\parind$. Therefore, the map $(z, \discountfactor_1, \cdots, \discountfactor_\parind) \mapsto \val(W_\discfac^\state(z))$ is continuous, which yields the item.
        
        \item 
        For all $\widehat{\strategyone} \in \Strategyone^{PS}$ and $\widehat{\strategytwo} \in \Strategytwo^{PS}$, we have
        \begin{align*}
            W_\discfac^\state(z_1)[\widehat{\strategyone}, \widehat{\strategytwo}] - W_\discfac^\state(z_2)[\widehat{\strategyone}, \widehat{\strategytwo}] &= (z_2 - z_1) \deter_\discfac(\widehat{\strategyone}, \widehat{\strategytwo}) & (\text{def. } W_\discfac^\state)\\
            &= (z_2 - z_1) \det \left (\id - \left ( (\mathbf{1} - \discfac)\mathbf{1}^\top \right ) \odot \prob^{\widehat{\strategyone}, \widehat{\strategytwo}} \right ) & \left (\text{def. } \deter_\discfac(\widehat{\strategyone}, \widehat{\strategytwo}) \right )\\
            &\ge (z_2 - z_1)\left (\min_i \discountfactor_i \right )^\statenum & (\text{\Cref{Result: Determinant of diagonally dominating matrix}})
        \end{align*}
        The result follows from the fact that increasing each entry of a matrix game by at least $t$ increases the value by at least $t$.
        
        \item  
        By the result of \cite{shapley1953stochastic}, there exist optimal stationary strategies $\strategyone^\ast$ and $\strategytwo^\ast$ for CSGs with stateful-discounted objectives. Therefore, for all $\widehat{\strategytwo} \in \Strategytwo^{PS}$, we have
        \[
            \payoff^{\strategyone^\ast, \widehat{\strategytwo}}(\state) = \frac{\deter_\discfac^\state(\strategyone^\ast, \widehat{\strategytwo})}{\deter_\discfac(\strategyone^\ast, \widehat{\strategytwo})} \ge \val_\discfac(\state) \,.
        \]
        Hence,
        \[
            \deter_\discfac^\state(\strategyone^\ast, \widehat{\strategytwo}) - \val_\discfac(\state) \cdot \deter_\discfac(\strategyone^\ast, \widehat{\strategytwo}) \ge 0 \,.
        \]
        By \Cref{Result: generalization of W}, in the matrix game $W_\discfac^\state(\val_\discfac(\state))$, we have
        \[
            W_\discfac^\state \left (\val_\discfac(\state) \right )[\bm{\strategyone}^\ast, \widehat{\strategytwo}] = \deter_\discfac^\state(\strategyone^\ast, \widehat{\strategytwo}) - \val_\discfac(\state) \cdot \deter_\discfac(\strategyone^\ast, \widehat{\strategytwo}) \ge 0 \quad \forall \widehat{\strategytwo} \in \Strategytwo^{PS} \,,
        \]
        which guarantees that $\val \left (W_\discfac^\state(\val_\discfac(\state)) \right ) \ge 0$. By symmetric arguments on $\strategytwo^*$, we get $\val \left (W_\discfac^\state(\val_\discfac(\state)) \right ) \le 0$. The result follows from combining these two inequalities.
    \end{enumerate}
\end{proof}

\begin{proof}[Proof of \Cref{Result: Characterization of discounted value}]
By \Cref{Result: Properties of W}-\Cref{Item: Monotonicity of W}, the mapping $z \mapsto \val \left ( W_\discfac^\state(z) \right )$ is strictly decreasing. By \Cref{Result: Properties of W}-\Cref{Item: W is 0 at discounted value}, we know that $\val \left ( W_\discfac^\state(\val_\discfac(\state)) \right ) = 0$. Hence, there exists the unique $z^* = \val_\discfac(\state) \in \RR$ such that 
\[
    \val \left (W_\discfac^\state(z^*) \right ) = 0 \,,
\]    
which yields the result.
\end{proof}

\subsection{Characterization of Limit Value}
\label{Section: Characterization of limit value}
In this subsection, we introduce a new characterization of the limit value in CSGs (\Cref{Result: Characterization of undiscounted value,Result: Properties of F}), which generalizes the result presented in \cite{attia2019formula} from a single discount factor to multiple discount factors. In particular, \Cref{Result: Characterization of undiscounted value} generalizes Theorem 2 of \cite{attia2019formula}.

\smallskip\noindent{\bf Limit function.} Given a CSG $G$, a reward function, and an assignment function $\assign \colon \States \to [\parind]$, we define the limit function as 
\[
    F_{\assign}^\state(z) \defas \lim_{\discountfactor_1 \to 0^+} \cdots \lim_{\discountfactor_\parind \to 0^+} \frac{\val \left( W^\state_\discfac(z) \right )}{(\discountfactor_\parind)^\statenum} \,.
\]

\begin{lemma}
\label{Result: Properties of F}
    Consider a CSG $G$, a state $\state$, a reward function, and an assignment function $\assign \colon \States \to [\parind]$. Then, the following assertions hold.
    \begin{enumerate}
        \item 
        \label{Item: Existence of F}
        For all $z \in \RR$, the limit $F_{\assign}^\state(z)$ exists in $\RR \cup \{ -\infty, +\infty \}$; and

        \item
        \label{Item: Non-constancy of F} There exists $z_1, z_2 \in \RR$ such that $F_{\assign}^\state(z_2) \le 0 \le F_\assign^\state(z_1)$. 
    \end{enumerate}
\end{lemma}

\begin{corollary}
\label{Result: Characterization of undiscounted value}
    Consider a CSG $G$, a state $\state$, a reward function, and an assignment function $\assign \colon \States \to [\parind]$. Then, $\val_{\assign}(\state)$ is the unique \(z^* \in \RR\) such that  
    \[
        \forall z > z^* \quad F_{\assign}^\state(z) < 0
        \qquad \text{and} \qquad
        \forall z < z^* \quad F_{\assign}^\state(z) > 0\,.
    \]    
\end{corollary}
Below we first present \Cref{Result: Existence of a rational for the limit of matrix game}, which shows that, given a fixed parameter $z$, the value of the matrix game is a rational function when discount factors are small enough, which implies the existence of the limit function. We then prove \Cref{Result: Properties of F} and \Cref{Result: Characterization of undiscounted value}.

\begin{lemma}
\label{Result: Existence of a rational for the limit of matrix game}
    Consider a CSG $G$, a state $\state$, a reward function, an assignment function $\assign \colon \States \to [\parind]$, and a real number $z \in \RR$. Then, there exist two polynomials $P$ and $Q$ in $\discountfactor_1, \cdots, \discountfactor_\parind$ such that
    \begin{gather*}
        \exists \discountfactor^0_1 > 0 \quad \forall \discountfactor_1 \in (0, \discountfactor^0_1) \quad \cdots \quad \exists \discountfactor^0_\parind > 0 \quad \forall \discountfactor_\parind \in (0, \discountfactor^0_\parind)\\
        \text{s.t. } \quad Q(\discountfactor_1, \cdots, \discountfactor_\parind) \neq 0 \quad \text{and} \quad \val \left ( W_\discfac^\state(z) \right ) = \frac{P(\discountfactor_1, \cdots, \discountfactor_\parind)}{Q(\discountfactor_1, \cdots, \discountfactor_\parind)} \, .
    \end{gather*}
\end{lemma}

\begin{proof}
    By \Cref{Result: Value of matrix games}, there exists a sub-matrix $M$ in $W_\discfac^\state(z)$ such that
    \[
        \val \left ( W_\discfac^\state(z) \right ) = \frac{\det(M)}{S(M)}\,.
    \]
    Since $z$ is fixed, the matrix $M$ is a matrix with polynomial entries in $\discountfactor_1, \cdots, \discountfactor_\parind$. Therefore, by \Cref{Result: Det of poly matrix,Result: S of poly matrix}, there exist two polynomials $P$ and $Q$ in $\discountfactor_1, \cdots, \discountfactor_\parind$ such that
    \[
        \val \left ( W_\discfac^\state(z) \right ) = \frac{P(\discountfactor_1, \cdots, \discountfactor_\parind)}{Q(\discountfactor_1, \cdots, \discountfactor_\parind)}\,.
    \]
    Since the matrix $W_\discfac^\state(z)$ is finite, the set of all sub-matrices is finite. Therefore, there exist two finite sets of polynomials $\setPolynomials$ and $\setPolynomialDenominators$ such that for all $\discountfactor_1, \cdots, \discountfactor_\parind$, there exists $P \in \setPolynomials$ and $Q \in \setPolynomialDenominators$ such that 
    \[
        \val \left ( W_\discfac^\state(z) \right ) = \frac{P(\discountfactor_1, \cdots, \discountfactor_\parind)}{Q(\discountfactor_1, \cdots, \discountfactor_\parind)} \,.
    \]
    By \Cref{Result: Properties of W}-\Cref{Item: Continuity of W}, we have $\val \left ( W_\discfac^\state(z) \right )$ is continuous in $\discountfactor_1, \cdots, \discountfactor_\parind$. Therefore, as $(\discountfactor_1, \cdots, \discountfactor_\parind)$ varies in $(0, 1]^\parind$, the value of the matrix game can only jump from one rational function to another if the graphs of two rationals intersect. More formally, $\val \left (W_\discfac^\state(z) \right )$ can jump from $\frac{P_1}{Q_1}$ to $\frac{P_2}{Q_2}$ when discount factors are $\discountfactor_1, \cdots, \discountfactor_\parind$, if we have 
    \[
        \frac{P_1(\discountfactor_1, \cdots, \discountfactor_\parind)}{Q_1(\discountfactor_1, \cdots, \discountfactor_\parind)} = \frac{P_2(\discountfactor_1, \cdots, \discountfactor_\parind)}{Q_2(\discountfactor_1, \cdots, \discountfactor_\parind)} \,,
    \] 
    We now claim that for every $P_1, P_2 \in \setPolynomials$ and $Q_1, Q_2 \in \setPolynomialDenominators$, either $\frac{P_1}{Q_1}$ and $\frac{P_2}{Q_2}$ are congruent, or
    \begin{gather*}
        \exists \discountfactor^0_1 > 0 \quad \forall \discountfactor_1 \in (0, \discountfactor^0_1) \quad \cdots \quad \exists \discountfactor^0_\parind > 0 \quad \forall \discountfactor_\parind \in (0, \discountfactor^0_\parind)\\
        \text{s.t. } \quad  \frac{P_1(\discountfactor_1, \cdots, \discountfactor_\parind)}{Q_1(\discountfactor_1, \cdots, \discountfactor_\parind)} \neq \frac{P_2(\discountfactor_1, \cdots, \discountfactor_\parind)}{Q_2(\discountfactor_1, \cdots, \discountfactor_\parind)} \,.
    \end{gather*}
    Indeed, we define the polynomial $C \defas P_1 \cdot Q_2 - P_2 \cdot Q_1$.
    
    \smallskip\noindent{\em Case $C = 0$.} If $C = 0$, then $\frac{P_1}{Q_1}$ and $\frac{P_2}{Q_2}$ are congruent, which complete the case.

    \smallskip\noindent{\em Case $C \neq 0$.} If $C \neq 0$, then by \Cref{Result: n-variate root separation from 0}, we have that 
    \begin{gather*}
        \exists \discountfactor^0_1 > 0 \quad \forall \discountfactor_1 \in (0, \discountfactor^0_1) \quad \cdots \quad \exists \discountfactor^0_\parind > 0 \quad \forall \discountfactor_\parind \in (0, \discountfactor^0_\parind)\\
        \text{s.t. } \quad C(\discountfactor_1, \cdots, \discountfactor_\parind) \neq 0 \,.
    \end{gather*}
    If $C(\discountfactor_1, \cdots, \discountfactor_\parind) \neq 0$, then $\frac{P_1(\discountfactor_1, \cdots, \discountfactor_\parind)}{Q_1(\discountfactor_1, \cdots, \discountfactor_\parind)} \neq \frac{P_2(\discountfactor_1, \cdots, \discountfactor_\parind)}{Q_2(\discountfactor_1, \cdots, \discountfactor_\parind)}$, which completes the case and concludes the proof. 
\end{proof}

\begin{proof}[Proof of \Cref{Result: Properties of F}]
    We prove the two items as follows.
    \begin{enumerate}
        \item It is a direct implication of \Cref{Result: Existence of a rational for the limit of matrix game}.
        
        \item
        We define $z_1 \defas \min \reward(\state, \actionone, \actiontwo)$ and $z_2 \defas \max \reward(\state, \actionone, \actiontwo)$. For all $\discountfactor_1, \cdots, \discountfactor_\parind$, we have
       \[
            \val \left ( W_\discfac^\state(z_2) \right ) 
                \le \val \left ( W_\discfac^\state(\val_\discfac(\state)) \right ) \le \val \left ( W_\discfac^\state(z_1) \right ) \,.
        \]
        By dividing both side by $(\discountfactor_\parind)^\statenum$ and taking $\discountfactor_\parind, \cdots, \discountfactor_1$ to 0 respectively, we get
        \[
            F_\assign^\state(z_2) \le 0 \le F_\assign^\state(z_1) \,,
        \]
        which yields the result.
    \end{enumerate}
\end{proof}

\begin{proof}[Proof of \Cref{Result: Characterization of undiscounted value}]
    Consider the function $F_\assign^\state$.
    By \Cref{Result: Properties of F}-\Cref{Item: Non-constancy of F}, it is not constant.
    By \Cref{Result: Properties of W}-\Cref{Item: Monotonicity of W}, it is decreasing.
    Let $z^*$ be the point where $z \mapsto F_\assign^\state(z)$ changes sign. 
    For all $\eps > 0$, we have that $F_\assign^\state(z^* + \eps) < 0$.
    Hence,
    \begin{gather*}
        \exists \discountfactor^0_1 > 0 \quad \forall \discountfactor_1 \in (0, \discountfactor^0_1) \quad \cdots \quad \exists \discountfactor^0_\parind > 0 \quad \forall \discountfactor_\parind \in (0, \discountfactor^0_\parind)\\
        \text{s.t.} \quad \val \left (W_\discfac^\state(z^* + \eps) \right ) < 0 \,.
    \end{gather*}
    By \Cref{Result: Properties of W}-\Cref{Item: Monotonicity of W}, 
    \begin{gather*}
        \exists \discountfactor^0_1 > 0 \quad \forall \discountfactor_1 \in (0, \discountfactor^0_1) \quad \cdots \quad \exists \discountfactor^0_\parind > 0 \quad \forall \discountfactor_\parind \in (0, \discountfactor^0_\parind)\\
        \text{s.t.} \quad \val_\discfac(\state) < z^* + \eps \,.
    \end{gather*}
    Therefore, for all $\eps > 0$, we get 
    \[
        \lim_{\discountfactor_1 \to 0^+} \cdots \lim_{\discountfactor_\parind \to 0^+} \val_\discfac(\state) \le z^* + \eps \,.
    \]
    By taking $\eps$ to 0, we have
    \[
        \lim_{\discountfactor_1 \to 0^+} \cdots \lim_{\discountfactor_\parind \to 0^+} \val_\discfac(\state) \le z^* \,.
    \]
    By symmetric arguments on $z^* - \eps$, we get
    \[
        \lim_{\discountfactor_1 \to 0^+} \cdots \lim_{\discountfactor_\parind \to 0^+} \val_\discfac(\state) \ge z^* \,,
    \]
    which concludes the proof.
\end{proof}

\subsection{Approximation of Limit Value}
\label{Section: Approximation of limit value}
In this subsection, we introduce an approach for the approximation of the limit value via the stateful-discounted value (\Cref{Result: Approximation of undiscounted value by discounted value}). 
The rest of the subsection is dedicated to its proof.

\begin{theorem}
\label{Result: Approximation of undiscounted value by discounted value}
    Consider a CSG $G$, a state $\state$, a reward function~$\reward$, an assignment function $\assign \colon \States \to [\parind]$, and an additive error $\eps > 0$. The transition function $\prob$ and the reward function~$\reward$ are represented by rational numbers of bit-size $\B$. 
    Fix 
    \[
        \degree \defas \max(|\Strategyone^{PS}|, |\Strategytwo^{PS}|), \quad
        \B_1 \defas 11\degree\statenum(\B + \bit(\statenum) + \bit(\degree) + \bit(\eps))\,,
    \]
    and, for all $1 \le i \le \parind$, we set $\discountfactor^0_i \defas \exp \left (-\B_1 (\statenum \degree+1)^{i-1} \right )$. 
    Then, we have
    \[
        |\val_{\discfac^0}(\state) - \val_\assign(\state)| \le \eps\,.
    \]
\end{theorem}

\begin{remark}[Novelty]
     As mentioned previously, our result is a generalization of~\cite{attia2019formula}. The key non-trivial aspect of the generalization relies on the fact that~\cite{attia2019formula} considers uni-variate polynomials, whereas our result requires analysis of multi-variate polynomials. \Cref{Result: n-variate root separation from 0} is the key mathematical foundation, and the complete proofs require significant technical generalization.
\end{remark}
Below we first show the algebraic properties of $\val(W_\discfac^\state(z))$ (\Cref{Result: Value of derived matrix game}) to derive some insights on the asymptotic behavior of the sign of the map $(\discountfactor_1, \cdots, \discountfactor_\parind) \mapsto \val(W_\discfac^\state(z))$ as $\discountfactor_\parind, \cdots, \discountfactor_1$ go to 0 respectively (\Cref{Result: Sign of F based on sign of W}). We then use \Cref{Result: Sign of F based on sign of W} to establish a connection between the stateful-discounted value and the limit value in \Cref{Result: Approximation of undiscounted value by discounted value}.

\begin{lemma}
\label{Result: Value of derived matrix game}
    Consider a CSG $G$, a state $\state$, a reward function $\reward$, and an assignment function $\assign \colon \States \to [\parind]$. The transition function $\prob$ and the reward function $\reward$ are represented by rational numbers of bit-size $\B$. 
    Let $\degree \defas \max(|\Strategyone^{PS}|, |\Strategytwo^{PS}|)$.
    Then, there exist two finite sets $\setPolynomials$ and~$\setPolynomialDenominators$ of nonzero polynomials in $z, \discountfactor_1, \cdots, \discountfactor_\parind$ of degrees $\degree, \statenum \degree \cdots, \statenum \degree$ with integer coefficients such that for all $z \in \RR$ and $\discountfactor_1, \cdots, \discountfactor_\parind$, there exist $P \in \setPolynomials$ and $Q \in \setPolynomialDenominators$ such that 
    \[
    Q(z, \discountfactor_1, \cdots, \discountfactor_\parind) \neq 0, \quad \val \left ( W_\discfac^\state(z) \right ) = \frac{P(z, \discountfactor_1, \cdots, \discountfactor_\parind)}{Q(z, \discountfactor_1, \cdots, \discountfactor_\parind)} \,.
    \]
    Moreover, the coefficients of $P$ are of bit-size $7\degree\statenum(\B + \bit(\statenum) + \bit(\degree))$.
\end{lemma}

\begin{proof}
    By definition, $W_\discfac^\state(z)[\widehat{\strategyone}, \widehat{\strategytwo}] = \deter_\discfac^\state(\widehat{\strategyone},\widehat{\strategytwo}) - z \cdot \deter_\discfac^\state(\widehat{\strategyone}, \widehat{\strategytwo})$ for all pure stationary strategies $\widehat{\strategyone}$ and $\widehat{\strategytwo}$. 
    Note that $\deter_\discfac^\state(\widehat{\strategyone}, \widehat{\strategytwo})$ and $\deter_\discfac(\widehat{\strategyone}, \widehat{\strategytwo})$ are the determinant of two matrices whose entries are polynomial in $\discountfactor_1, \cdots, \discountfactor_\parind$ of degrees $1, 1, \cdots, 1$ with coefficients from the set $\{0, 2^{-B}, 2 \cdot 2^{-\B}, \cdots, 1\}$. 
    Therefore, by \Cref{Result: Det of poly matrix}, the entries of $2^{\statenum\B}W_\discfac^\state(z)$ are polynomials in $z$, $\discountfactor_1, \cdots, \discountfactor_\parind$ of degrees $1, \statenum, \cdots, \statenum$ with integer coefficients of bit-size 
    \[
        \statenum\B + \statenum\bit(\statenum) + \parind \bit(\statenum+ 1) \le \statenum\B + 3\statenum\bit(\statenum) \sadef \B'\,. 
    \]
    We now define two sets of nonzero polynomials $\setPolynomials$ and $\setPolynomialDenominators$ as 
    \[
        \setPolynomials \defas \left \{  \det(M) \mid M \text{ is a sub-matrix of } 2^{\statenum\B}W_\discfac^\state(z) \right \}\,,
    \]
    \[
        \setPolynomialDenominators \defas \left \{  2^{\statenum\B}S(M) \mid M \text{ is a sub-matrix of } 2^{\statenum\B}W_\discfac^\state(z) \right \}\,.
    \]
    By \Cref{Result: Value of matrix games}, for all $z$, there exists a sub-matrix $M$ in the matrix game $2^{\statenum\B}W_\discfac^\state(z)$ such that 
    \[
        \val \left( W_\discfac^\state \right) = \frac{\det(M)}{2^{\statenum\B}S(M)}  \,,
    \]
    where the equality is due to the fact that $\det(2^{-\statenum\B}M) = 2^{-k\statenum\B}\det(M)$ and $S(2^{-\statenum\B}M) = 2^{-(k-1)\statenum\B}S(M)$ for the $k \times k$~matrix $M$. Note that $M$ is a polynomial matrix. Hence, there exist $P \in \setPolynomials$ and $Q \in \setPolynomialDenominators$ such that for all $z$ and $\discountfactor_1, \cdots, \discountfactor_\parind$, we have
    \[
        Q(z, \discountfactor_1, \cdots, \discountfactor_\parind) \neq 0, \quad \val \left( W_\discfac^\state \right) = \frac{P(z, \discountfactor_1, \cdots, \discountfactor_\parind)}{Q(z, \discountfactor_1, \cdots, \discountfactor_\parind)} \,.
    \]
    We now show that $\setPolynomials$ and $\setPolynomialDenominators$ satisfy the requirements of the lemma. 
    Let $P \defas \det(M)$ and $Q  \defas 2^{\statenum\B}S(M)$, where $M$ is a $k \times k$ sub-matrix of $2^{\statenum\B}W_\discfac^\state(z)$. First, by \Cref{Result: Properties of W}-\Cref{Item: Monotonicity of W}, we know that $z \mapsto \val(W_\discfac^\state(z))$ is strictly decreasing, therefore $P$ is nonzero. Second, We know that entries of $M$ are polynomials in $z, \discountfactor_1, \cdots, \discountfactor_\parind$ of degrees $1, \statenum, \cdots, \statenum$ with integer coefficients of bit-size $\B'$. Therefore, by \Cref{Result: Det of poly matrix}, we have that $P$ is a polynomial in $z, \discountfactor_1, \cdots, \discountfactor_\parind$ of degrees $k, k \statenum, \cdots, k \statenum$ with integer coefficients of bit-size 
    \begin{align*}
        k\B' + k \bit(k) + (\parind+1) \bit(k \statenum + 1) &\le \degree\B' + \degree \bit(\degree) + (\parind+1) \bit(\statenum \degree + 1) \\
        &\le \degree(\B' + \bit(\degree) + 4\statenum\bit(\statenum \degree))\\
        &= \degree(\statenum\B + 3\statenum \bit(\statenum) + \bit(\degree) + 4\statenum\bit(\statenum \degree)) \\
        &\le \degree \statenum(\B + 3\bit(\statenum) + \bit(\degree) + 4\bit(\statenum \degree)) \\
        &\le 7\degree\statenum(\B + \bit(\statenum) + \bit(\degree))\,,
    \end{align*}
    where in the first inequality we use $k \le \degree$, in the second inequality we use $(\parind+1)\bit(\statenum \degree+1) \le 4\statenum \degree \bit(\statenum \degree)$, in the first equality we use $B' = \statenum\B + 3\statenum \bit(\statenum)$, in the third inequality we use $\bit(\degree) \le \statenum\bit(\degree)$, and in the fourth inequality we use $\bit(\statenum \degree) \le \bit(\statenum) + \bit(\degree)$. 
    
    Similarly, since $S(M)$ is the sum of the entries of the adjugate matrix of $M$, by \Cref{Result: S of poly matrix}, we have that $Q$ is a polynomial in $z, \discountfactor_1, \cdots, \discountfactor_\parind$ of degrees $\degree, \statenum \degree, \cdots, \statenum \degree$ with integer coefficients, which completes the proof.
\end{proof}

\begin{lemma}
\label{Result: Sign of F based on sign of W}
    Consider a CSG $G$, a state $\state$, a reward function $\reward$, an assignment function $\assign \colon \States \to [\parind]$, and a rational number $z$ of bit-size $\kappa$. The transition function $\prob$ and the reward function $\reward$ are represented by rational numbers of bit-size $\B$. 
    Fix 
    \[
        \degree \defas \max \left( \left| \Strategyone^{PS} \right|, \quad  
        \left|\Strategytwo^{PS} \right| \right), \quad 
        \B_1 \defas 11\degree\statenum(\B + \bit(\statenum) + \bit(\degree) + \kappa)\,,
    \]
    and, for all $1 \le i \le \parind$, we have $\discountfactor^0_i \defas \exp \left (-\B_1 (\statenum \degree+1)^{i-1} \right )$. 
    Then, 
    \begin{gather*}
        \begin{cases}
        \val \left ( W_{\discfac^0}^\state(z) \right ) > 0 &\implies F_{\assign}^\state(z) \ge 0\,,\\
        \val \left ( W_{\discfac^0}^\state(z) \right ) < 0 &\implies F_{\assign}^\state(z) \le 0\,,\\
        \val \left ( W_{\discfac^0}^\state(z) \right ) = 0 &\implies F_{\assign}^\state(z) = 0\,.
        \end{cases}
    \end{gather*}
\end{lemma}
\begin{proof}
    We claim that 
    \begin{gather*}
        \forall \discountfactor_1 \in \left (0, \exp(-\B_1) \right ] \quad \cdots \quad \forall \discountfactor_i \in \left (0, (\discountfactor_{i-1})^{\statenum \degree + 1} \right] \quad \cdots \quad \forall \discountfactor_\parind \in \left (0, (\discountfactor_{\parind-1})^{\statenum \degree + 1} \right]\\
        \sign \left ( \val ( W_\discfac^\state(z)) \right ) = \sign ( \val (W_{\discfac^0}^\state(z))) \,,
    \end{gather*}
    which proves the theorem. By \Cref{Result: Properties of W}-\Cref{Item: Continuity of W}, the map $(z, \discountfactor_1, \cdots, \discountfactor_\parind) \mapsto \val(W_\discfac^\state(z))$ is continuous. Therefore, the necessary condition for changing sign of $(\discountfactor_1, \cdots, \discountfactor_\parind) \mapsto \val (W_\discfac^\state(z))$ is that $\val( W_{\discfac^1}^\state(z)) = 0$ for some $\discountfactor^1_1, \cdots, \discountfactor^1_\parind$. For the sake of contradiction, assume there exists $\discountfactor^1_1, \cdots, \discountfactor^1_\parind$ such that $\discountfactor^1_1 \le \exp(-\B_1)$ and $\discountfactor^1_i \le (\discountfactor^1_{i-1})^{\statenum \degree+1}$ for all $i > 1$. 
    Let $\setPolynomials$ and $\setPolynomialDenominators$ be the finite sets defined in \Cref{Result: Value of derived matrix game}. 
    Let $P \in \setPolynomials$ and $Q \in \setPolynomialDenominators$ such that 
    \[
        \val \left (W_{\discfac^1}^\state(z) \right ) = \frac{P(z, \discountfactor
        ^1_1, \cdots, \discountfactor_\parind^1)}{Q(z, \discountfactor^1_1, \cdots, \discountfactor^1_\parind)} \,.
    \]
    Since $\val (W_{\discfac^1}^\state(z)) = 0$, we have that $P(z, \discountfactor^1_1, \cdots, \discountfactor^1_\parind) = 0$. 
    Let $z$ be fixed. 
    We define
    \[
        P_z(\discountfactor_1, \cdots, \discountfactor_\parind) 
            \defas \exp(\kappa d) \cdot P(z, \discountfactor_1, \cdots, \discountfactor_\parind) \,.
    \]
    Polynomial $P$ is a polynomial in $z, \discountfactor_1, \cdots, \discountfactor_\parind$ of degrees $\degree, \statenum \degree, \cdots, \statenum \degree$ with integer coefficients of bit-size $7\degree\statenum(\B + \bit(\statenum) + \bit(\degree))$. Since $z$ is a rational number of bit-size $\kappa$, we have that $P_z$ is a polynomial in $\discountfactor_1, \cdots, \discountfactor_\parind$ of degrees $\statenum \degree, \cdots, \statenum \degree$ with integer coefficients of bit-size $7\degree\statenum(\B + \bit(\statenum) + \bit(\degree) + \kappa)$. If $P_z = 0$, then on the one hand, we have $\val \left (W_{\discfac^0}^\state(z) \right ) = 0$, and on the other hand, we have that 
    \begin{gather*}
        \forall \discountfactor_1 \in \left (0, \exp(-\B_1) \right ] \quad \cdots \quad \forall \discountfactor_i \in \left (0, (\discountfactor_{i-1})^{\statenum \degree + 1} \right] \quad \cdots \quad \forall \discountfactor_\parind \in \left (0, (\discountfactor_{\parind-1})^{\statenum \degree + 1} \right]\\
         \val ( W_\discfac^\state(z)) = 0 \,,
    \end{gather*}
    which implies that $F_{\assign}^\state(z) = 0$ and the proof of this case is done. Therefore, $P_z$ is nonzero. Since $\B_1 = 11\degree\statenum(\B + \bit(\statenum) + \bit(\degree) + \kappa)$, by \Cref{Result: n-variate root separation from 0}, we have that $P_z(\discountfactor^1_1, \cdots, \discountfactor^1_\parind) \neq 0$, which contradicts with the assumption and completes the proof.
\end{proof}

\begin{proof}[Proof of \Cref{Result: Approximation of undiscounted value by discounted value}]
    We set $\kappa \defas \bit(\eps)$. 
    We define $\mathcal{Z} \defas \{0, 2^{-\kappa}, 2 \cdot 2^{-\kappa}, \cdots, 1\}$. 
    By \Cref{Result: Properties of W}-\Cref{Item: Monotonicity of W}, we have that $z \mapsto \val(W_{\discfac^0}^\state(z))$ is strictly decreasing. 
    Therefore, there exists $z \in \mathcal{Z}$ such that $\val(W_{\discfac^0}^\state(z)) \ge 0$ and $\val(W_{\discfac^0}^\state(z + 2^{-\kappa})) \le 0$. By \Cref{Result: Properties of W}-\Cref{Item: Monotonicity of W}, we have
    \begin{equation}
    \label{Eq: Bound on the discounted value}
        z \le \val_{\discfac^0}(\state) \le z + 2^{-\kappa}\,.
    \end{equation}
    By \Cref{Result: Sign of F based on sign of W} and \Cref{Result: Characterization of undiscounted value}, we have
    \begin{equation}
    \label{Eq: Bound on the undiscounted value}    
        z \le \val_\assign(\state) \le z + 2^{-\kappa}\,.
    \end{equation}
    The result follows from combining \Cref{Eq: Bound on the discounted value} and \Cref{Eq: Bound on the undiscounted value}.
\end{proof}

%% file: algorithms.tex
\section{Algorithms for \UndiscVal\ and \ParVal}
\label{Section: Algorithms}
In this section, we present algorithms for computing $\eps$-approximation of stateful-discounted, limit, and parity values. The section is organized as follows. In \Cref{Section: Selected algos}, we recall two classical algorithmic procedures that are used in our algorithms. In \Cref{Section: Algorithm for discounted value}, we present an algorithm for computing $\eps$-approximate stateful-discounted value. In \Cref{Section: Algorithm for undiscounted value}, we present an algorithm for computing $\eps$-approximate limit value, and as a consequence, we also obtain an algorithm for computing $\eps$-approximate parity value. 

\subsection{Selected Algorithms from Literature}
\label{Section: Selected algos}
In this subsection, we recall classical algorithms for computing the determinant of a matrix and computing the value of a matrix game.

\begin{lemma}[\cite{bareiss1968sylvester}]
\label{Result: Computing Det of a matrix}
Consider a $k \times k$ matrix $M$ with rational entries of bit-size $\B$. Then, there exists a procedure \Call{\Det}{$M$} that computes the determinant of $M$ in time $\tilO \left (k^4 \B^2 \right )$, and \Call{\Det}{$M$} is of bit-size $\calO(k\log(k)\B)$.
\end{lemma}

\begin{lemma}[\cite{karmarkar1984new}]
\label{Result: Computing Value of a matrix game}
    Consider a $k \times k$ matrix game $M$ with rational entries of bit-size $\B$. 
    Then, there exists a procedure \Call{\Val}{$M$} that computes the value of the matrix game $M$ in time $\tilO(k^{3.5}\B)$, and \Call{\Val}{$M$} is of bit-size $\calO(\B)$.
\end{lemma}

\subsection{Algorithm for Approximate Stateful-discounted Value}
\label{Section: Algorithm for discounted value}

In this subsection, we present an algorithm for computing $\eps$-approximation of the stateful-discounted value in CSGs. 
Given a CSG $G$, a reward function, and a discount function~$\discfac$, the procedure runs a binary search over the stateful-discounted value of state $\state$. 
At the beginning, $\underline{z}$ and $\overline{z}$ are the under and over approximation of $\val_\discfac(\state)$. In each step, the algorithm halves the interval $[\underline{z}, \overline{z}]$ by increasing $\underline{z}$ or decreasing $\overline{z}$ based on the sign of $\val \left (W_\discfac^\state \left (\frac{\underline{z} + \overline{z}}{2} \right ) \right )$. 
After $\bit(\eps)$ steps, the algorithm outputs the $\eps$-approximate value $(\overline{z} + \underline{z})/2$. The formal description is shown in \Cref{algorithm: approx discounted}, and the correctness and the time complexity of the algorithm are shown in \Cref{Result: approx discounted algorithm}.
\begin{algorithm}[ht]
  \caption{\ApproxDiscounted}
  \label{algorithm: approx discounted}
  \begin{algorithmic}[1]
    \Require Game $G$, state $\state$, reward function $\reward$, a discount function $\discfac$, additive error $\eps$
    \Ensure Approximate stateful-discounted value $v$ such that  $|v - \val_\discfac(\state)| \le \eps$
 
    \Procedure{\ApproxDiscounted}{$G, \state, \reward, \discfac, \eps$}
      \State $\underline{z} \gets 0$ and $\overline{z} \gets 1$
      \While{$\overline{z} - \underline{z} > \eps$}
        \State $z \gets \frac{\underline{z} + \overline{z}}{2}$
        \State $\nu \gets \val(W_\discfac^\state(z))$
        \label{Line: Matrix Game Computations}
        \If{$\nu \ge 0$}
            \State $\underline{z} \gets z$
        \Else
            \State $\overline{z} \gets z$
        \EndIf
      \EndWhile
      \State \Return $\frac{\underline{z} + \overline{z}}{2}$
    \EndProcedure
  \end{algorithmic}
\end{algorithm}

\begin{lemma}
\label{Result: approx discounted algorithm}
    Consider a CSG $G$, a state $\state$, a reward function $\reward$, a discount function $\discfac$, and an additive error $\eps > 0$. The transition function $\prob$, the reward function $\reward$, and the discount function $\discfac$ are represented by rational numbers of bit-size $\B$. Then, \Cref{algorithm: approx discounted} computes the $\eps$-approximation of the stateful-discounted value of state $\state$. Moreover, the algorithm runs in time $\exp \left (\calO(\statenum\log(\actionnum) + \log(B) + \log(\log(1/\eps)) \right )$.
\end{lemma}

\begin{proof}
We first present the proof of correctness and then the time complexity of the algorithm.

    \smallskip\noindent{\em Correctness.} The procedure is a binary search over the stateful-discounted value of state $\state$. At the beginning, $\underline{z}$ and $\overline{z}$ are the under and over approximation of $\val_\discfac(\state)$. In each step, by \Cref{Result: Properties of W}-\Cref{Item: Monotonicity of W}, the algorithm halves the interval $[\underline{z}, \overline{z}]$ by increasing (resp. decreasing) $\underline{z}$ (resp.~$\overline{z}$) based on the sign of $\nu$. Therefore, in all steps, $\underline{z} \le \val_\discfac(\state) \le \overline{z}$ is invariant. The algorithm terminates after at most $\bit(\eps)$ steps. The correctness of the procedure follows from the invariant $\underline{z} \le \val_\discfac(\state) \le \overline{z}$ and $\overline{z} - \underline{z} \le \eps$. 
    
    \smallskip\noindent{\em Time complexity.} The procedure executes at most $\bit(\eps)$ iterations. All lines except Line~\ref{Line: Matrix Game Computations} require constant arithmetic operations. Line~\ref{Line: Matrix Game Computations} consists of two parts as follows. 
    \begin{itemize}
        \item {\em Construction of $W_\discfac^\state(z)$.} Recall that 
        \[
            W_\discfac^\state(z)[\widehat{\strategyone}, \widehat{\strategytwo}] = \deter_\discfac^\state(\widehat{\strategyone}, \widehat{\strategytwo}) - z \cdot \deter_\discfac(\widehat{\strategyone}, \widehat{\strategytwo})\,,
        \]
        where $\deter_\discfac^\state(\widehat{\strategyone}, \widehat{\strategytwo})$ and $\deter_\discfac(\widehat{\strategyone}, \widehat{\strategytwo})$ are the determinants of two $\statenum \times\statenum$ matrices. 
        The construction of matrices runs in time $\calO(\statenum^2\B)$. 
        The algorithm uses \Call{\Det}{} to compute the determinants, which runs in $\tilO(\statenum^4\B^2)$ by \Cref{Result: Computing Det of a matrix}. 
        The determinants are of bit-size $\calO(\statenum\log(\statenum)\B)$. 
        The number of entries of $W_\discfac^\state(z)$ is at most $\actionnum^{2\statenum}$. Therefore, its construction runs in time $\tilO(\actionnum^{2\statenum})$. 
        The entries of $W_\discfac^\state(z)$ is of bit-size $\calO(\statenum\log(\statenum)\B + \log(1/\eps))$.
        \item {\em Computation of $\val(W_\discfac^\state(z))$.} 
        The procedure uses \Call{\Val}{} to compute the value of the matrix game $W_\discfac^\state(z)$. Therefore, by \Cref{Result: Computing Value of a matrix game}, it runs in time $\tilO \left (\actionnum^{3.5\statenum}\B^2 \log^2(1/\eps) \right )$.
    \end{itemize}
    Hence, \Cref{Line: Matrix Game Computations} runs in time $\tilO(\actionnum^{3.5\statenum}\B^2\log^3(1/\eps))$, which completes the proof.

\end{proof}

\subsection{Algorithms for Approximate Limit and Parity Values}
\label{Section: Algorithm for undiscounted value}
In this subsection, we present an algorithm for computing $\eps$-approximation of the limit and parity values in CSGs. Given a CSG $G$, a reward function, and an assignment function~$\assign$, the procedure outputs the $\eps/2$-approximate of the stateful-discounted value of state $\state$ for some $\discfac_0$ by calling \Call{\ApproxDiscounted}{}. By \Cref{Result: Approximation of undiscounted value by discounted value}, the stateful-discounted value is an $\eps/2$-approximation of the limit value. Thus, the returned value of the algorithm is indeed an $\eps$-approximate of the limit value. The formal description is shown in \Cref{algorithm: approx limit}, and the correctness and the time complexity of the algorithm is shown in \Cref{Result: approx limit algorithm}. Since CSGs with parity objectives have a linear-size reduction to CSGs with the limit value of stateful-discounted objectives, as a consequence of the above algorithm, we obtain an algorithm for parity value approximation.

\begin{algorithm}[ht]
  \caption{\ApproxLimit}
  \label{algorithm: approx limit}
  \begin{algorithmic}[1]
    \Require Game $G$, state $\state$, reward function $\reward$, assignment function $\assign$, additive error $\eps$
    \Ensure Approximate limit value $v$ such that  $|v - \val(\state)| \le \eps$
 
    \Procedure{\ApproxLimit}{$G, \state, \reward, \assign, \eps$}
        \State $\degree \gets \actionnum^\statenum$
        \State $\B_1 \gets 11\degree\statenum \left (\B + \bit(\statenum) + \bit(\degree) + \bit(\eps) \right )$
        \For{$i \gets 1$ to $\parind$}
            \State $\discountfactor^0_i \gets \exp \left (-\B_1(\statenum \degree+1)^{i-1} \right )$
        \EndFor
        \State 
        \label{Line: Proc call of discounted value}
        $v \gets \Call{\ApproxDiscounted}{G, s, \reward, \discfac^0, \eps/2}$
        \State \Return $v$
    \EndProcedure
  \end{algorithmic}
\end{algorithm}

\begin{lemma}
\label{Result: approx limit algorithm}
    Consider a CSG $G$, a state $\state$, a reward function $\reward$, an assignment function $\assign\colon \States \to [\parind]$, and an additive error $\eps > 0$. The transition function $\prob$ and the reward function $\reward$ are represented by rational numbers of bit-size $\B$. Then, \Cref{algorithm: approx limit} computes the $\eps$-approximation of the limit value of state $\state$. Moreover, the algorithm runs in time $\exp\left( \calO(\statenum\parind \log(\actionnum) + \log(\B) + \log(\log(1/\eps))) \right)$.
\end{lemma}

\begin{proof}
    We first present the proof of correctness and then the time complexity of the algorithm. 

    \smallskip\noindent{\em Correctness.} The procedure computes the $\eps/2$-approximate stateful-discounted value $v$ for $\discfac^0$ and outputs $v$ as the approximate limit value. The procedure \Call{\ApproxDiscounted}{} outputs $v$ such that 
    \begin{equation}
    \label{Eq: closeness of nu and discounted value}    
    \left |v - \val_{\discfac^0}(\state) \right | \le \eps/2\,.
    \end{equation}
    By \Cref{Result: Approximation of undiscounted value by discounted value}, we have that 
    \begin{equation}
    \label{Eq: closeness of undiscounted value and discounted value}    
    \left |\val_{\discfac^0}(\state) - \val_\assign(\state) \right | \le \eps/2\,.
    \end{equation}
    By combining \Cref{Eq: closeness of nu and discounted value,Eq: closeness of undiscounted value and discounted value}, we get that $v$ is $\eps$-approximation of $\val_\assign(\state)$.
    
    \smallskip\noindent{\em Time complexity.}
    All lines except Line~\ref{Line: Proc call of discounted value} require at most $\calO(\parind)$ arithmetic operations. In \Cref{Line: Proc call of discounted value}, the algorithm calls \Call{\ApproxDiscounted}{} with parameters $(G, \state, \reward, \discfac^0, \eps/2)$. The bit-size of the discount function $\discfac^0$ is $\tilO \left (\statenum^\parind \actionnum^{\statenum\parind} (\B + \bit(\eps)) \right )$. Therefore, this line runs in time $\tilO(\actionnum^{7.5\statenum\parind} \B^2 \log^5(1/\eps))$, which completes the proof.
\end{proof}

\begin{proof}[Proof of {\Cref{Result: Computational result for limit value}-\Cref{Item: Algorithmic time for limit value}}]
    It is a direct implication of \Cref{Result: approx limit algorithm}.
\end{proof}

\begin{corollary}
\label{Result: approx parity algorithm}
    Consider a CSG $G$, a priority function $p$, a state $\state$, and an additive error $\eps > 0$. The transition function $\prob$ is represented by rational numbers of bit-size $\B$. Then, there exists an algorithm that computes the $\eps$-approximation of the parity value of state $\state$. Moreover, the algorithm runs in time $\exp \left (\calO(\statenum\parind \log(\actionnum) + \log(\B) + \log(\log(1/\eps))) \right )$.
\end{corollary}

\begin{proof}
    By \cite{gimbert2005discounting,de2003discounting}, there exists a linear-size reduction from the CSGs with parity objectives to the CSGs with the limit-value of stateful-discounted objectives. Therefore, the result follows from \Cref{Result: approx limit algorithm}.
\end{proof}

\begin{proof}[Proof of {\Cref{Result: Computational result for parity value}-\Cref{Item: Algorithmic time for parity value}}]
    It is a direct implication of \Cref{Result: approx parity algorithm}.
\end{proof}

%% file: complexity.tex
\section{Complexities of \UndiscVal\ and \ParVal}
\label{Section: Complexities}
In this section, we show that the \UndiscVal\ and \ParVal\ problems are in TFNP[NP]. 
This section is organized as follows.
In \Cref{Section: Complexity preliminaries}, we present some useful definitions and selected results from the literature related to Markov Chains (MCs) and Markov Decision Processes (MDPs), and floating-point representation.
In \Cref{Section: Complexity of computing stateful-discounted value in MCs,Section: Complexity of computing stateful-discounted value in MDPs,Section: Complexity of computing stateful-discounted value in CSGs}, we present algorithms for computing $\eps$-approximate stateful-discounted value in MCs, MDPs, and CSGs. 
Each algorithm is used in the subsequent algorithm as a procedure. By our technical result on the limit value approximation via the stateful-discounted value (\Cref{Result: Approximation of undiscounted value by discounted value}), we consequently obtain a TFNP[NP] procedure for the \UndiscVal\ problem. Since there exists a linear-size reduction from CSGs with parity objectives to CSGs with the limit-value of stateful-discounted objectives~\cite{gimbert2005discounting,de2003discounting}, the \ParVal\ problem is also in TFNP[NP].

\subsection{Definitions and Selected Results from Literature}
\label{Section: Complexity preliminaries}
We present some basic notations and definitions related to Markov Chains, Markov Decision Processes, and the classic symbolic representation for numbers and probability distributions, called floating-point. 

\smallskip\noindent{\bf Markov decision processes and Markov chains.}
For $i \in \{1, 2\}$, a Player-$i$ Markov decision process (Player-$i$ MDP) is a special class of CSGs where the other player has only one action and is denoted by $\MDP = (\States, \Actionone, \prob \colon \States \times \Actionone \to \Delta(\States))$.
A Markov chain (MC) is a special class of MDPs where both players have only one action and is denoted by $\MC = (\States, \prob \colon \States \to \Delta(\States))$. In Markov chains we write $\prob(s,s')$ to denote $\prob(s)(s')$.

\smallskip\noindent{\bf Absorbing MCs.} 
We say an MC $\MC$ is \emph{absorbing} if there exists a subset of absorbing states $\States_0 \subseteq \States$ such that
\begin{itemize}
    \item For all $\state \in \States_0$, we have $\prob(\state, \state) = 1$; and
    \item For all $\state_0 \in \States \setminus \States_0$, there exist states $\state_1, \ldots, \state_k$ such that $\prob(\state_i, \state_{i+1}) > 0$ and $\state_k \in \States_0$.
\end{itemize}
States in $\States_0$ are called absorbing.

\smallskip\noindent{\bf MDPs and MCs given stationary strategies in CSGs.}
Given a stationary strategy $\strategyone$ for Player~1 in a game $G$, by fixing the strategy $\strategyone$, we obtain a Player-2 MDP $G_\strategyone = (S, \Actiontwo, \prob_\strategyone)$ 
where the transition function $\prob_\strategyone \colon \States \times \Actiontwo \to \Delta(\States)$ is given by
\[
    \prob_\strategyone(\state, \actiontwo)(\state') \defas \sum_{\actionone \in \Actionone} \prob(\state, \actionone, \actiontwo)(s') \cdot \strategyone(\state)(\actionone)\,,
\]
for all $\state, \state' \in \States$ and $\actiontwo \in \Actiontwo$. Analogously, we obtain Player-1 MDP $G_\strategytwo$ by fixing a stationary strategy $\strategytwo$ for Player~2. 
Moreover, by fixing stationary strategies $\strategyone$ and $\strategytwo$ for both players, we obtain an MC $G_{\strategyone, \strategytwo} = (\States, \prob_{\strategyone, \strategytwo})$, where the transition function $\prob_{\strategyone, \strategytwo} \colon \States \to \Delta(\States)$ is given by 
\[
    \prob_{\strategyone, \strategytwo}(\state)(\state') = \sum_{\actionone \in \Actionone} \sum_{\actiontwo \in \Actiontwo} \prob(\state, \actionone, \actiontwo)(\state') \cdot \strategyone(\state)(\actionone) \cdot \strategytwo(\state)(\actiontwo) \,,
\]
for all $\state, \state' \in \States$.

\smallskip\noindent{\bf Reachability objectives in MCs.} Given an MC $\MC$ and a target set $\Target \subseteq \States$, the reachability objective is the indicator function of plays eventually reaching $\Target$. More formally, for a play $\play = \langle \state_0, \state_1, \cdots \rangle$, we define $\reach_\Target \colon \Plays \to \{0, 1\}$ as
\[
    \reach_\Target(\play) \defas \begin{cases}
        1 & \exists i \ge 0 \; \state_i \in \Target \\
        0\,. & 
    \end{cases}
\]
We define the probability of reaching the target set $\Target$ from state $\state$ as $\val_\Target(\state) \defas \EE_\state[\reach_\Target]$.

\smallskip\noindent{\bf Floating-point number representation.} We define the set of floating-point numbers with precision $\ell$ as 
\begin{align*}
  \F(\ell) 
    \defas \left\{ m \cdot 2^{e} \quad \mid \quad m \in \{0, \cdots, 2^\ell-1\}, \quad e \in \ZZ \right \} \,.
\end{align*}
The floating-point representation of an element $x = m \cdot 2^e \in \F(\ell)$ uses $\bit(m) + \bit(|e|)$ bits. We define the relative distance of two positive real numbers $x, \widetilde{x}$ as 
\[
  \rel(x, \widetilde{x}) 
    \defas \max \left\{ \frac{x}{\widetilde{x}}, \frac{\widetilde{x}}{x} \right\} - 1 
    = \inf \{ \alpha > 0 \enspace : \enspace x \le (1 + \alpha) \widetilde{x}, \quad \widetilde{x} \le (1 + \alpha) x \}\,.
\]
We say $x$ is $(\ell, i)$-close to $\widetilde{x}$ if 
    $\rel(x, \widetilde{x}) \le (1 - 2^{1 - \ell})^{-i} - 1$,
where $\ell$ is a positive integer and $i$ is a non-negative integer.

\smallskip\noindent{\bf Arithmetic operations.} We define $\oplus^{\ell}, \ominus^{\ell}, \otimes^{\ell}, \oslash^{\ell}$ as finite precision arithmetic operations $+, -, *, /$ respectively by truncating the result of the exact arithmetic operation to $\ell$ bits. We drop the superscript $\ell$ if context is clear.

\smallskip\noindent{\bf Floating-point probability distribution representation.}
We denote by $\D(\ell)$ the set of all floating-point probability distributions with precision $\ell$. 
A probability distribution $\distribution \in \Delta([t])$ belongs to $\D(\ell)$ if there exists $w_1, w_2, \cdots, w_t \in \F(\ell)$ such that
\begin{itemize}
    \item For all $i \in [t]$, we have $\distribution(i) = \frac{w_i}{\sum_{j \in [t]} w_j}$; and
    \item $\sum_{j \in [t]} w_j$ and 1 are $(\ell, t)$-close.
\end{itemize}
We define the relative distance $\rel$ for probability distributions as
$
    \rel (\distribution, \widetilde{\distribution}) \defas \max \{ \rel(\distribution(i), \widetilde{\distribution}(i)) : i \in [t]\}
$. 
We say $\distribution$ is $(\ell, i)$-close to $\widetilde{\distribution}$ if 
    $\rel(\distribution, \widetilde{\distribution}) \le (1 - 2^{1 - \ell})^{-i} - 1$,
where $\ell$ is a positive integer and $i$ is a non-negative integer.

Below we recall some useful results from the literature related to MCs, MDPs, and the floating-point representation.

\begin{lemma}[\protect{\cite[Thm. 6]{solan2003continuity}}]
\label{Result: Continuity of value in MC with reachability}
    Consider two absorbing MCs $\MC$ and $\widetilde{\MC}$ with identical state sets and a target set $T$. 
    We denote by $\val_T$ and $\widetilde{\val}_T$ the reachability value of $\MC$ and $\widetilde{\MC}$ respectively. 
    Fix $\eps \defas \max_{\state, \state'} \rel(\prob(\state, \state'), \widetilde{\prob}(\state, \state'))$. Then, for all states $\state \in \States$, we have
    \[
        |\val_T(\state) - \widetilde{\val}_T(\state)| \le 4 \statenum \eps \,.
    \]
\end{lemma}

\begin{lemma}[\protect{\cite[Thm. 4]{frederiksen2013approximating}}]
\label{Result: Approximate reachability value}
    Consider an absorbing MC $\MC$ and a target set $T$. 
    For all $\state \in \States$, we have $\prob(\state) \in \D(\ell)$ where $\ell \ge 1000 \statenum^2$. 
    Then, there exists a polynomial-time algorithm that for all states $\state \in \States$, computes an approximation $v \in \F(\ell)$ for the reachability value such that 
    \[
        |v - \val_T(\state)| \le 80 \statenum^4 2^{-\ell}\,.
    \]
\end{lemma}

\begin{lemma}[\protect{\cite[Lemma 1]{frederiksen2013approximating}}]
\label{Result: Transitivity of rel}
    Consider three non-negative real numbers $x, y, z$. If $x$ and $y$ are $(\ell, i)$-close, and $y$ and $z$ are $(\ell, j)$-close, then $x$ and $z$ are $(\ell, i+j)$-close.
\end{lemma}

\begin{lemma}[\protect{\cite[Lemma 4]{frederiksen2013approximating}}]
\label{Result: FP arithmetic operations}
Consider non-negative real numbers $x, y$. 
Let $\widetilde{x} \in \F(\ell)$ be a number that is $(\ell, i)$-close to $x$ and $\widetilde{y} \in \F(\ell)$ be a number that is $(\ell, j)$-close to $y$. 
Then, the following assertions hold.

\begin{enumerate}
    \item The number $\widetilde{x} \oplus \widetilde{y}$ is $(\ell, \max(i, j) + 1)$-close to $x + y$;
    \item The number $\widetilde{x} \ominus \widetilde{y}$ is $(\ell, \max(i, j) + 1)$-close to $x - y$;
    \item The number $\widetilde{x} \otimes \widetilde{y}$ is $(\ell, i + j + 1)$-close to $x * y$; and
    \item The number $\widetilde{x} \oslash \widetilde{y}$ is $(\ell, i + j + 1)$-close to $x / y$.
\end{enumerate}
Moreover, all arithmetic operations can be computed in polynomial time with respect to $\ell$.
\end{lemma}

\begin{lemma}[\protect{\cite[Lemma 5]{frederiksen2013approximating}}]
\label{Result: Approximation of prob distribution}
    Consider $x_1, \cdots, x_t \in \F(\ell)$. 
    Let $\distribution(i) \defas x_i \oslash \left ( \bigoplus_{j=1}^t x_j \right )$. 
    Then, there exists $\widetilde{\distribution} \in \D(\ell)$ such that for all $i$, we have $\widetilde{\distribution}(i) = \distribution(i) / 
 \left (\sum_{j=1}^{t} \distribution(j) \right )$, and $\distribution$ and $\widetilde{\distribution}$ are $(\ell, 2t)$-close. 
\end{lemma}

\begin{lemma}[\protect{\cite[Lemma 6]{frederiksen2013approximating}}]
\label{Result: Approximation of prob distribution from distribution}
    Consider a probability distribution $\distribution \in \Delta([t])$. Then, there exists $\widetilde{\distribution} \in \D(\ell)$ such that $\distribution$ and  $\widetilde{\distribution}$ are $(\ell, 2t+2)$-close.
\end{lemma}

\subsection{Stateful-discounted Value Approximation in MCs}
\label{Section: Complexity of computing stateful-discounted value in MCs}
In this subsection, we present an algorithm for computing $\eps$-approximate stateful-discounted value in MCs by a reduction from  MCs with stateful-discounted objectives to MCs with reachability objectives.

\begin{lemma}
\label{Result: Approximate discounted value in MC}
    Consider an MC $\MC$, a reward function $\reward$, and a discount function $\discfac$. 
    For all $\state \in \States$,  we set 
    \[
        \prob(\state) \in D(\ell), \quad \reward(\state) \in \F(\ell), \quad \discfac(\state) \in \F(\ell) \,, 
    \]
    where $\ell \ge 1000 \statenum^2$. Then, there exists a polynomial-time algorithm that for all states $\state \in \States$, computes an approximation $v$ for the  stateful-discounted value such that 
    \[
        |v - \val_\discfac(\state)| \le 104 \statenum^4 2^{-\ell}\,.
    \]
\end{lemma}

\begin{proof}
    We construct a new MC $\MC_1$ from $\MC$ with a reachability objective and two additional absorbing states $\top$ and $\bot$ so the set of states is $\States_1 \defas \States \cup \{\top, \bot\}$. 
    The target set is $T \defas \{ \top \}$, and the state $\bot$ is absorbing. 
    The transition function $\prob_1$ is defined as
    \[
        \prob_1(\state, \state') \defas \begin{cases}
            (1 - \discfac(\state)) \cdot \prob(\state, \state') & \state, \state' \in \States\\
            \sum_{\state'' \in \States} \discfac(\state) \cdot \reward(\state) \cdot \prob(\state, \state'') & \state \in \States, \; \state' = \top\\
            \sum_{\state'' \in \States} \discfac(\state) \cdot (1 - \reward(\state)) \cdot \prob(\state, \state'') & \state \in \States,\; \state = \bot\\
            1 & \state \in \{\top, \bot\},\; \state' = \state\\ 
            0 & \text{otherwise}
        \end{cases}
    \]
    We denote by $\val_\discfac$ the stateful-discounted value in $\MC$ and by $\val^1_T$ the reachability value in $\MC_1$. 
    Note that $\MC_1$ is absorbing. 
    Observe that the Bellman equation for $\MC_1$ with the stateful-discounted objective is the same as the Bellman equation for $\MC_1$ with the reachability objective. 
    Therefore, for all $\state \in \States$, we have
    \begin{equation}
    \label{Eq: Equality of reachability and discounted value}
        \val_\discfac(\state) = \val^1_T(\state)\,.
    \end{equation}
    We define the approximate transition function as 
    \[
        \prob_2(\state, \state') \defas \begin{cases}
            (1 \ominus \discfac(\state)) \otimes \prob(\state, \state') & \state, \state' \in \States\\
            \bigoplus_{\state'' \in \States} \discfac(\state) \otimes \reward(\state) \otimes \prob(\state, \state'') & \state \in \States, \; \state' = \top\\
            \bigoplus_{\state'' \in \States} \discfac(\state) \otimes (1 \ominus \reward(\state)) \otimes \prob(\state, \state'') & \state \in \States,\; \state = \bot\\
            1 & \state = \{\top, \bot\},\; \state' = \state\\ 
            0 & \text{otherwise}
        \end{cases}
    \]
    By the definition of finite precision arithmetic operators, we have $\prob_2(\state, \state') \in \F(\ell)$. 
    Also, by \Cref{Result: FP arithmetic operations}, we have $\prob_1(\state, \state)$ and $\prob_2(\state, \state')$ are $(\ell, \statenum + 3)$-close, and $\prob_2$ is computable in polynomial time. ‌By \Cref{Result: Approximation of prob distribution}, there exists $\prob_3 \in \D(\ell)$ such that $\prob_2$ and $\prob_3$ are $(\ell, 2 \statenum)$-close. 
    Moreover, $\prob_3$ is computable in polynomial time. Therefore, by \Cref{Result: Transitivity of rel}, we have $\prob_1$ and $\prob_3$ are $(\ell, 3 \statenum + 3)$-close. 
    Therefore, for all $\state, \state' \in \States$, we have
    \begin{align*}
        \rel(\prob_1(\state, \state'), \prob_3(\state, \state')) &\le \frac{1}{(1 - 2^{1 - \ell})^{3 \statenum + 3}} - 1\\
        &\le \frac{1}{1 - (3 \statenum + 3)2^{1 - \ell}} - 1 & (\text{Bernoulli inequality})\\
        &\le \frac{(3 \statenum + 3)2^{1 - \ell}}{1 - (3 \statenum + 3)2^{1 - \ell}} & (\text{rearrange})\\
        &\le (3 \statenum + 3)2^{-\ell} & \left (\ell \ge 1000 \statenum^2 \right )\\
        &\le 6 \statenum 2^{-\ell}\,. & (\statenum \ge 1)
    \end{align*}
    We define an MC $\MC_2$ derived from $\MC_1$ with transition function $\prob_3$. 
    We denote by $\val^2_T$ the reachability value in $\MC_2$. 
    By \Cref{Result: Continuity of value in MC with reachability}, for all $\state \in \States$, we have
    \begin{equation}
    \label{Eq: closeness of two reachability value}
        |\val^1_T(\state) - \val^2_T(\state)| \le 24 \statenum^2 2^{-\ell}\,.
    \end{equation}
    By \Cref{Result: Approximate reachability value}, for all $\state \in \States$, we can compute an approximation $v(\state)$ of the reachability value of $\MC_2$ starting from $\state$ in polynomial time such that
    \begin{equation}
    \label{Eq: approximation of reachability value}
        |v(\state) - \val^2_T(\state)| \le 80 \statenum^4 2^{-\ell}\,.
    \end{equation}
    By combining \Cref{Eq: Equality of reachability and discounted value,Eq: closeness of two reachability value,Eq: approximation of reachability value}, for all states $\state \in \States$, we have
    \[
         |v(\state) - \val_\discfac(\state)| \le 104 \statenum^4 2^{-\ell}\,,
    \]
    which yields the result.
\end{proof}

\subsection{Stateful-discounted Value Approximation in MDPs}
\label{Section: Complexity of computing stateful-discounted value in MDPs}
In this subsection, we present an NP procedure for the approximate decision problem of the stateful-discounted values in MDPs. 
This is accomplished by guessing a pure stationary strategy and verifying if the strategy achieves the given threshold. 
The verification procedure uses the algorithm for computing approximate stateful-discounted values in MCs.

\begin{lemma}
\label{Result: Approximate discounted value in player-1 MDP}
    The problem of deciding if the stateful-discounted value for Player-1 MDPs is below a threshold up to an additive error is in NP where the input is a Player-1 MDP $\MDP$, a reward function $\reward$, a discount function $\discfac$, a state $\state$, a threshold $0 \le \alpha \le 1$, an additive error $\eps = 2^{-\kappa}$ and a positive integer $\ell$ such that, for all $\state' \in \States$ and $\actionone \in \Actionone$, we have 
    \[
        \prob(\state', \actionone) \in D(\ell), \quad \reward(\state', \actionone) \in \F(\ell), \quad \discfac(\state') \in \F(\ell) , \quad
        \ell \ge 1000 \statenum^2 + \kappa \,.
    \]
    Note that, the numbers $\alpha$ and $\eps$ are represented in fixed-point binary and the NP procedure is such that 
    \begin{itemize}
        \item If $\alpha \le \val_\discfac(\state) - \eps$, then it outputs YES; and
        \item If $\alpha \ge \val_\discfac(\state) + \eps$, then it outputs NO.
    \end{itemize}
\end{lemma}



\begin{proof}
    We first present the NP procedure and then prove its soundness and completeness.

    \smallskip\noindent{\em Procedure.} 
    The procedure guesses a pure stationary strategy $\strategyone$ for Player~1. 
    Note that the size of the representation of a pure stationary strategy is polynomial with respect to the size of the representation of $\MDP$. 
    By fixing $\strategyone$, we obtain an MC $\MDP_\strategyone$. 
    We denote by $v_\strategyone$ its stateful-discounted value. 
    By \Cref{Result: Approximate reachability value}, there exists a polynomial time algorithm that computes an $\eps$-approximation $\widetilde{v}_\strategyone$ of $v_\strategyone$. 
    Our procedure outputs YES if $\alpha \le \widetilde{v}_\strategyone(\state)$. 
    If there exists no such pure stationary strategy, the procedure outputs NO.
    
    \smallskip\noindent{\em Soundness.} 
    If $\alpha \le \val_\discfac(\state) - \eps$, then, by \cite{FV97}, there exists a pure stationary strategy $\strategyone$ such that $\val_\discfac = v_\strategyone$. 
    The procedure non-deterministically guesses $\strategyone$. 
    By \Cref{Result: Approximate discounted value in MC}, we have $\widetilde{v}_\strategyone(\state) + \eps \ge v_\strategyone(\state)$. 
    Therefore, we have $\alpha \le \widetilde{v}_\strategyone(\state)$, and the procedure outputs YES.

    \smallskip\noindent{\em Completeness.} 
    If $\alpha \ge \val_\discfac(\state) + \eps$, then for all pure stationary strategies $\strategyone$, we have $\alpha \ge v_\strategyone(\state) + \eps$. 
    By \Cref{Result: Approximate discounted value in MC}, we have $\widetilde{v}_\strategyone(\state) - \eps \le v_\strategyone(\state)$. 
    Therefore, $\alpha \ge \widetilde{v}_\strategyone(\state)$ which implies that the procedure outputs NO and yields the result.
\end{proof}

\Cref{Result: Approximate discounted value in player-1 MDP} also holds for Player-2 MDPs by symmetric arguments. 
More formally, we have the following result.

\begin{corollary}
\label{Result: Approximate discounted value in player-2 MDP}
    The problem of deciding if the stateful-discounted value for Player-2 MDPs is above a threshold up to an additive error is in NP where the input is a Player-2 MDP $\MDP$, a reward function $\reward$, a discount function $\discfac$, a state $\state$, a threshold $0 \le \alpha \le 1$, an additive error $\eps = 2^{-\kappa}$ and a positive integer $\ell$ such that, for all $\state' \in \States$ and $\actionone \in \Actionone$, we have 
    \[
        \prob(\state', \actionone) \in D(\ell), \quad \reward(\state', \actionone) \in \F(\ell), \quad \discfac(\state') \in \F(\ell) , \quad
        \ell \ge 1000 \statenum^2 + \kappa \,.
    \]
\end{corollary}



\subsection{Stateful-discounted, Limit, and Parity Values Approximation in CSGs}
\label{Section: Complexity of computing stateful-discounted value in CSGs}

In this subsection, we first present a continuity result for MDPs with the stateful-discounted objectives (\Cref{Result: discounted MDP continuity}), which is a generalization of~\cite[Eq. 4.19]{FV97}. 
We then show the existence of $\eps$-optimal stationary strategies for stateful-discounted objectives that are representable in polynomial-size with respect to the size of the game and $\bit(\eps)$ (\Cref{Result: Patience upper bound}). 
We finally present a TFNP[NP] algorithm for CSGs with the stateful-discounted objectives, where the discount function is represented in floating-point (\Cref{Result: DiscVal problem is in TFNP[NP]}). 
By our technical result on the limit value approximation via the stateful-discounted value (\Cref{Result: Approximation of undiscounted value by discounted value}), we consequently obtain a TFNP[NP] procedure for the \UndiscVal\ problem. Since there exists a linear-size reduction from CSGs with parity objectives to CSGs with the limit-value of stateful-discounted objectives~\cite{gimbert2005discounting,de2003discounting}, the \ParVal\ problem is also in TFNP[NP].

\begin{lemma}
    \label{Result: discounted MDP continuity}
    Consider $i \in \{1, 2\}$ and two Player-$i$ MDPs $\MDP$ and $\widetilde{\MDP}$ with identical state and action sets, a reward function $\reward$, and a discount function $\discfac$. 
    We denote by $\val_\discfac$ and $\widetilde{\val}_\discfac$ the stateful-discounted value of $\MDP$ and $\widetilde{\MDP}$ respectively. 
    Then, for all $\state \in \States$, we have
    \[
        \left| \val_\discfac(\state) - \widetilde{\val}_\discfac(\state) \right| 
            \le \frac{\| \prob - \widetilde{\prob} \|_\infty}{\min_{\state} \discfac(\state)}  \,.
    \]
\end{lemma}
\begin{proof}
    We only prove for Player-1 MDPs. The proof for Player-2 MDPs is symmetric. By the Bellman equation defined in \cite{shapley1953stochastic}, for all $\state \in \States$, we have
    \begin{align*}
        & \left| \val_\discfac(\state) - \widetilde{\val}_\discfac(\state) \right| \\
        &\quad \le \max_{\actionone} \Bigg | 
            \discfac(\state) \reward(\state, \actionone) 
            + (1 - \discfac(\state)) \sum_{\state'} \prob(\state, \actionone)(\state') \val_\discfac(\state') \\
            &\quad \phantom{\le \max_{\actionone} \Bigg |} \qquad 
            - \discfac(\state) \reward(\state, \actionone) 
            - (1 - \discfac(\state))\sum_{\state'} \widetilde{\prob}(\state, \actionone)(\state') \widetilde{\val}_\discfac(\state')  
        \Bigg | \\
        &\quad = (1 - \discfac(\state)) \max_\actionone \Bigg | 
            \sum_{\state'} \prob(\state, \actionone)(\state') (\val_\discfac(\state') 
            - \widetilde{\val}_\discfac(\state')) \\
            &\quad \phantom{= (1 - \discfac(\state)) \max_\actionone \Bigg |} \qquad 
            + \sum_{\state'} (\prob(\state, \actionone)(\state') 
            - \widetilde{\prob}(\state, \actionone)(\state')) \widetilde{\val}_\discfac(\state') 
        \Bigg |\\
        &\quad \le (1 - \discfac(\state)) \left ( \|\val_\discfac - \widetilde{\val}_\discfac\|_\infty + \|\prob - \widetilde{\prob}\|_\infty \, \|\widetilde{\val}_\discfac\|_\infty \right)\\
        &\quad \le (1 - \discfac(\state)) \left( \|\val_\discfac - \widetilde{\val}_\discfac\|_\infty + \|\prob - \widetilde{\prob}\|_\infty \right)\\
        &\quad \le \left( 1 - \min_{\state} \discfac(\state) \right) \left( \|\val_\discfac - \widetilde{\val}_\discfac\|_\infty + \|\prob - \widetilde{\prob}\|_\infty \right) \,,
    \stepcounter{equation}\tag{\theequation}\label{Eq: Recursive inequality for the difference of values}
    \end{align*}
    where in the first inequality we use the Bellman equation, in the first equality we use algebraic manipulation, in the second inequality we use the definition of the norm and Cauchy-Schwarz inequality, in the third inequality we use $\|\widetilde{\val}_\discfac\|_\infty \le \|\reward\|_\infty \le 1$, and in the fourth inequality we use $\min_\state \discfac(\state) \le \discfac(\state)$. 
    Since \Cref{Eq: Recursive inequality for the difference of values} holds for all states $\state$, we have 
    \[
        \|\val_\discfac - \widetilde{\val}_\discfac\|_\infty 
            \le \left( 1 - \min_{\state} \discfac(\state) \right) \left( \|\val_\discfac - \widetilde{\val}_\discfac\|_\infty + \|\prob - \widetilde{\prob}\|_\infty \right) \,,
    \]
    or equivalently
    \[
        \|\val_\discfac - \widetilde{\val}_\discfac\|_\infty 
        \le \left( \frac{1 - \min_{\state} \discfac(\state)}{\min_{\state} \discfac(\state)} \right) \|\prob - \widetilde{\prob}\|_\infty \le \frac{\|\prob - \widetilde{\prob}\|_\infty}{\min_{\state} \discfac(\state)} \,,
    \]
    which yields the result.
    
\end{proof}

\begin{lemma}
\label{Result: Patience upper bound}
    Consider a CSG $G$, a reward function, a discount function $\discfac$, and an additive error $\eps = 2^{-\kappa}$. 
    Fix $\underline{\Lambda} \defas \min_\state \discfac(\state)$. 
    Let $\ell$ be a positive integer such that $\ell \ge 4 \statenum \kappa \bit(\underline{\Lambda})$. 
    Then, there exist $\eps$-optimal stationary strategies $\strategyone$ and $\strategytwo$ such that
    \begin{align*}
        &\forall \state \in \States& 
            &\strategyone(\state) \in \D(\ell), \quad \strategytwo(\state) \in \D(\ell) \,, \\
        &\forall \state \in \States, \actionone \in \Actionone& 
            &\strategyone(\state)(\actionone) \neq 0 
                \implies \strategyone(\state)(\actionone) \ge \frac{\underline{\Lambda} \, \eps}{4} \,, \\
        &\forall \state \in \States, \actiontwo \in \Actiontwo& 
            &\strategytwo(\state)(\actiontwo) \neq 0 
                \implies \strategytwo(\state)(\actiontwo) \ge \frac{\underline{\Lambda} \, \eps}{4} \,.
    \end{align*}
\end{lemma}
\begin{proof}
    We only prove the existence of $\strategyone$ since the existence of $\strategytwo$ follows by symmetric arguments.
    By \cite{shapley1953stochastic}, there exists the optimal stationary strategy $\strategyone^*$ for Player 1. 
    We define $\strategyone_1$ as 
    \[
        \strategyone_1(\state)(\actionone) \defas 
        \begin{cases}
            0 & \strategyone^*(\state)(\actionone) \le  \frac{\underline{\Lambda}\eps}{2}\\
            \strategyone^*(\state)(\actionone) & \text{otherwise}
        \end{cases}
    \]
    Then, we have $\|\strategyone_1 - \strategyone^*\|_\infty \le \underline{\Lambda}\eps/2$.
    By \Cref{Result: Approximation of prob distribution from distribution}, there exists $\strategyone_2$ such that for all $\state \in \States$, we have $\strategyone_2(\state) \in \D(\States)$, and $\strategyone_1(\state)$ and $\strategyone_2(\state)$ are $(\ell, 2 \statenum + 2)$-close. Therefore, 
    \begin{align*}
        \|\strategyone_1 - \strategyone_2\|_\infty &\le \max_\state \; \rel(\strategyone_1(\state), \strategyone_2(\state))\\
        &\le \frac{1}{(1 - 2^{1 - \ell})^{2\statenum + 2}} - 1
            & (\text{def. $(\ell, 2 \statenum + 2)$-close})\\
        &\le \frac{1}{1 - (2\statenum+2)2^{1 - \ell}} - 1 
            & (\text{Bernoulli inequality})\\
        &\le \frac{(2\statenum+2)2^{1 - \ell}}{1 - (2\statenum+2)2^{1 - \ell}} 
            & (\text{rearrange})\\
        &\le (2\statenum+2)2^{-\ell} 
            & \left (\ell \ge 1000\statenum^2 \right )\\
        &\le 4\statenum2^{-\ell} 
            & (\statenum \ge 1)\\ 
        &\le \frac{\underline{\Lambda}\eps}{4} \,. 
            & (\ell \ge 4 \statenum \kappa \bit(\underline{\Lambda}))
    \end{align*}
    Observe that for all states $\state$ and actions $\actionone$, if $\strategyone_2(\state)(\actionone) \neq 0$, then $\strategyone_2(\state)(\actionone) \ge \underline{\Lambda}\eps/4$. 
    By combining two inequalities $\|\strategyone_1 - \strategyone^*\|_\infty \le \underline{\Lambda}\eps/2$ and $\|\strategyone_1 - \strategyone_2\|_\infty \le \underline{\Lambda}\eps/4$, we get $\|\strategyone_2 - \strategyone^*\|_\infty \le \underline{\Lambda} \eps$. 
    Therefore, by \Cref{Result: discounted MDP continuity}, we have $\strategyone_2$ is an $\eps$-optimal strategy, completing the proof.
\end{proof}

\begin{remark}
\label{Result: Polynomial-size strategies}
    For all CSGs with stateful-discounted objectives, there exist $\eps$-optimal stationary strategies which can be represented in polynomial size with respect to the size of the game and $\bit(\eps)$.
\end{remark}

\begin{lemma}
\label{Result: DiscVal problem is in TFNP[NP]}
    The problem of computing an $\eps$-approximation of the stateful-discounted value for CSGs is in TFNP[NP] for inputs CSGs $G$, reward functions, discount functions~$\discfac$, states~$\state$, additive errors $\eps = 2^{-\kappa}$, and positive integers $\ell$ such that, for all states $\state' \in \States$, we have $\discfac(\state') \in \F(\ell)$.
\end{lemma}
\begin{proof}
    We first present the procedure and then prove its soundness and completeness.

    \smallskip\noindent{\em Procedure.} The procedure guesses two $\eps/4$-optimal stationary strategies $\strategyone^*$ and $\strategytwo^*$ and an integer $j \in [0, 2^{\kappa+2}]$. 
    By \Cref{Result: Polynomial-size strategies}, the size of the representation of stationary strategies is polynomial with respect to the size of $G$ and $\bit(\eps)$. 
    By fixing strategy $\strategyone^*$ (resp. $\strategytwo^*)$, we obtain a Player-2 MDP $G_{\strategyone^*}$ (resp. a Player-1 MDP $G_{\strategytwo^*}$). 
    We denote by $v_{\strategyone^*}$ (resp. $v_{\strategytwo^*}$) the value of $G_{\strategyone^*}$ (resp. $G_{\strategytwo^*}$). 
    We define $\alpha \defas j 2^{-(\kappa + 2)} \in [0, 1]$. 
    The procedure outputs $\alpha$ if $\alpha - 3\eps/4 \le v_{\strategyone^*} - \eps/4$ and $\alpha + 3\eps/4 \ge v_{\strategytwo^*} + \eps/4$ which can be verified by two NP oracles implemented by \Cref{Result: Approximate discounted value in player-1 MDP,Result: Approximate discounted value in player-2 MDP}.
    
    \smallskip\noindent{\em Soundness.} 
    We assume that $\alpha$ is not an $\eps$-approximation of the stateful-discounted value of~$G$, i.e., that $\alpha \notin [\val_\discfac(\state) - \eps, \val_\discfac(\state) + \eps]$, and prove that the procedure does not output $\alpha$ in this case. 
    Without loss of generality, we assume $\alpha \ge \val_\discfac(\state) + \eps$. For all stationary strategies $\strategyone$, we have that $v_\strategyone \le \val_\discfac(\state)$. Therefore, we have
    \[
        \alpha - 3\eps/4 >\val_\discfac(\state) + \eps/4 \ge v_\strategyone + \eps/4\,.
    \]
    Hence, the NP procedure defined in \Cref{Result: Approximate discounted value in player-2 MDP} successfully decides that the inequality $\alpha - 3\eps/4 \le v_\strategyone - \eps/4$ is not true, and therefore the procedure does not output $\alpha$, which yields the soundness of the procedure.
    
    \smallskip\noindent{\em Completeness.} By \Cref{Result: Polynomial-size strategies}, there exist $\eps/4$-optimal stationary strategies $\strategyone^*$ and $\strategytwo^*$ which are polynomial-size representable. 
    The procedure non-deterministically guesses $\strategyone^*$ and $\strategytwo^*$. 
    Since both strategies are $\eps/4$-optimal, we have $\val_\discfac(\state) \le v_{\strategyone^*} + \eps/4$ and $v_{\strategytwo^*} - \eps/4 \le \val_\discfac(\state)$. By the choice of $j$, there exists $\alpha$ such that $\alpha \in [\val_\discfac(\state) - \eps/4, \val_\discfac(\state) + \eps/4]$. Therefore,  
    \begin{gather*}
        \alpha - 3\eps/4 \le \val_\discfac(\state) - \eps/2 \le v_{\strategyone^*} - \eps/4 \,, \text{ and} \\
        \alpha + 3\eps/4 \ge \val_\discfac(\state) + \eps/2 \ge v_{\strategytwo^*} + \eps/4\,,
    \end{gather*}
    and the procedure outputs $\alpha$.
    This yields the completeness of the procedure and completes the proof.
\end{proof}

\begin{proof}[Proof of {\Cref{Result: Computational result for limit value}-\Cref{Item: Complexity class for limit value}}]
    It is a direct implication of \Cref{Result: DiscVal problem is in TFNP[NP]} and \Cref{Result: Approximation of undiscounted value by discounted value}.
\end{proof}

\begin{proof}[Proof of {\Cref{Result: Computational result for parity value}-\Cref{Item: Complexity class for parity value}}]
    By \cite{gimbert2005discounting,de2003discounting}, there exists a linear-size reduction from the \ParVal\ problem to the \UndiscVal\ problem. 
    Therefore, the result follows from \Cref{Result: Computational result for limit value}-\Cref{Item: Complexity class for limit value}.
\end{proof}

%% file: conclusion.tex
\smallskip\noindent{\bf Concluding remarks.}
In this work, we present improved complexity upper bounds and algorithms for the value approximation problem for concurrent stochastic games
with two classic objectives. 
There are several interesting directions for future work. 
First, whether the complexity can be further improved from
TFNP[NP] to TFNP is a major open question, even for reachability objectives.
Second, whether for parity objectives, the dependency on $d$ can be improved from linear to logarithmic, retaining the logarithmic dependence on $m$, is another interesting open question.
Finally, the study of priority mean-payoff objectives for concurrent stochastic games and their connection to stateful-discounted objectives is another interesting direction for future work.